\documentclass[pdflatex,sn-mathphys-num]{sn-jnl}

\usepackage{amsmath, amssymb, mathtools}
\usepackage{amsthm} 

\usepackage{graphicx}
\usepackage[table]{xcolor}
\usepackage{array}
\usepackage{booktabs}
\usepackage{multirow}
\usepackage{makecell}
\usepackage{adjustbox}
\usepackage{calc}
\usepackage{colortbl}    
\usepackage{hyperref}

\usepackage{algorithm}
\usepackage{algpseudocode}
\usepackage{listings}

\usepackage{tikz}
\usetikzlibrary{arrows.meta, positioning, matrix, fit, backgrounds, shadings}
\pgfdeclarelayer{background}
\pgfsetlayers{background,main,foreground}

\usepackage{subcaption} 

\usepackage{textcomp}

\usepackage{cleveref} 

\newcommand{\defword}[1]{{\em #1}}
\newcommand{\equals}{:=}

\newtheorem{theorem}{Theorem}[section]
\newtheorem{lemma}{Lemma}[section]
\newtheorem{definition}{Definition}[section]

\newtheorem{corollary}{Corollary}[section]

\newtheorem{example}{Example}[section]

\theoremstyle{remark}
\newtheorem{remark}{Remark}[section]

\crefname{hypothesis}{Hypothesis}{Hypotheses}
\crefname{fact}{Fact}{Facts}

\makeatletter
\newcommand{\hyperfixedlabel}[2]{%
  \hypertarget{#1}{}
  \protected@write\@auxout{}{\string\newlabel{#1}{{#2}{\thepage}{}{#1}{}}}
}
\makeatother

\newcounter{matnum}[table]
\renewcommand{\thematnum}{\thetable.\arabic{matnum}}

\newcommand{\matref}[1]{\hyperlink{rommat:#1}{(\thetable.#1)}}

\usepackage{xr-hyper}
\usepackage{hyperref} 

\definecolor{sand}{RGB}{245,240,215}
\definecolor{palelilac}{RGB}{230,225,238}
\definecolor{taupe}{RGB}{225,215,205}
\definecolor{sagegreen}{RGB}{205,225,205}
\definecolor{coralorange}{RGB}{255,209,188}
\definecolor{LightPink}{HTML}{FFDEE9}
\definecolor{MintGreen}{HTML}{a1ffce}
\definecolor{lightblue}{RGB}{204,255,255}
\definecolor{lightgreen}{RGB}{210,245,210}
\definecolor{lightpurple}{RGB}{218, 112, 214}

\definecolor{nodeSingle}{RGB}{68,114,196}    
\definecolor{nodeMix}{RGB}{255,242,204}     
\definecolor{branchBg}{RGB}{226,240,217}    
\definecolor{branchEdge}{RGB}{146,208,80}   
\begin{document}

\title[Article Title]{Obstruction of Absolute Concentration Robustness by Conservation Laws in Non-Redundant Zero-One Networks}


\author[1]{\fnm{Xinyi} \sur{Si}}
\email{\text{sta\_sxy@buaa.edu.cn}}

\author*[1]{\fnm{Xiaoxian} \sur{Tang}}\email{xiaoxian@buaa.edu.cn}

\affil[1]{\orgdiv{School of Mathematical Sciences}, \orgname{Beihang University}, \orgaddress{
\postcode{100191}, \state{Beijing}, \country{China}}}


\abstract{Absolute concentration robustness (ACR) is a structural property of biochemical reaction networks in which a species attains the same steady-state concentration at every positive steady state, independently of initial conditions and rate constants. Existing detection methods rely on algebraic elimination and typically scale exponentially with network size. We develop a topology-based alternative for non-redundant zero-one networks of stoichiometric dimension at most two, a class that already captures enzyme catalysis, carbon-nanotube transitions, and other elementary biochemical mechanisms. Organizing our analysis around a structural index $s^*$, the number of distinct rows in the stoichiometric matrix, we obtain a complete classification of all such networks admitting non-vacuous ACR for generic rate constants. In dimension one, ACR occurs only for the elementary inflow and outflow module. In dimension two, ACR is possible if and only if $s^*\leq 3$; for $s^*=3$, the admissible networks are precisely those obtained as species refinements of consistent subnetworks of five canonical biochemical prototypes. For $s^*\geq 4$, non-vacuous ACR is impossible for any generic rate assignment.
}




\maketitle

\section{Introduction}

Mass-action kinetics provides a standard framework for modeling biochemical reaction networks \cite{HoJa1972}. Within this setting, absolute concentration robustness (ACR) identifies species whose steady-state concentrations are invariant across all positive steady states and independent of initial conditions and rate constants \cite{Shinar2010}. ACR is therefore a key indicator of structural robustness in cellular signaling and metabolic regulation, and has been studied from algebraic, dynamical, and control-theoretic perspectives \cite{Joshi2021, Joshi2023a, Joshi2023b, Yi2000, Briat2016, Cappelletti2019}.

Most existing ACR detection methods reduce the problem to analyzing steady-state polynomial systems via computational algebraic geometry \cite{Karp2012, PerezMillan2015, Puente2025, Baldi2021}. Although effective in moderate dimensions, such approaches---including Gr\"{o}bner basis elimination and real radical computation---generally scale exponentially with network size \cite{Schenck2003} and can fail when robust concentrations are irrational.

A rank-based criterion in \cite[Theorem 4.4]{Feliu2024}, built from the stoichiometric and reactant matrices, gives a necessary condition for ACR. When a network is not full-dimensional and is therefore subject to conservation laws, this necessary condition admits a more direct structural formulation, recorded here as an obstruction principle  (Theorem~\ref{thm:main}). The theorem considers removing a single species from a network while keeping the stoichiometric dimension unchanged. If the resulting network is nondegenerate, then the removed species cannot exhibit non-vacuous ACR for generic rate constants. Keeping the dimension unchanged upon removal means that the concentration of the removed species is coupled to the remaining species through a conservation law. Theorem~\ref{thm:main} therefore makes explicit an obstruction principle: conservation law structure, together with nondegeneracy of the resulting network, forbids ACR in the removed species. Although Theorem~\ref{thm:main} is not the main classification result of this paper, it serves a dual purpose beyond its later use in proofs. It explains why conservation laws are natural suspects in the search for structural obstructions to ACR, and it motivates the central question we address: how does conservation law structure determine when ACR can nevertheless occur?

Many biochemically studied ACR examples are naturally formulated as networks of elementary reactions with binary stoichiometric coefficients and no species appearing on both sides of any reaction. Fig.~\ref{fig:networks_background} collects representative models from the literature, including Phospho-OmpR dephosphorylation \cite[Fig.~2D]{Shinar2010}, the IDHKP--IDH glyoxylate bypass \cite[Example 3.3]{Joshi2024}, and a signal transduction module \cite[Example 4.17]{Puente2025}. Although these networks differ in scale and biological context, they share the same zero-one, non-redundant structure and therefore motivate our focus on this class as the appropriate setting for a complete structural classification of ACR.

\begin{figure}[!tbh]
    \centering
    \includegraphics[width=0.75\linewidth]{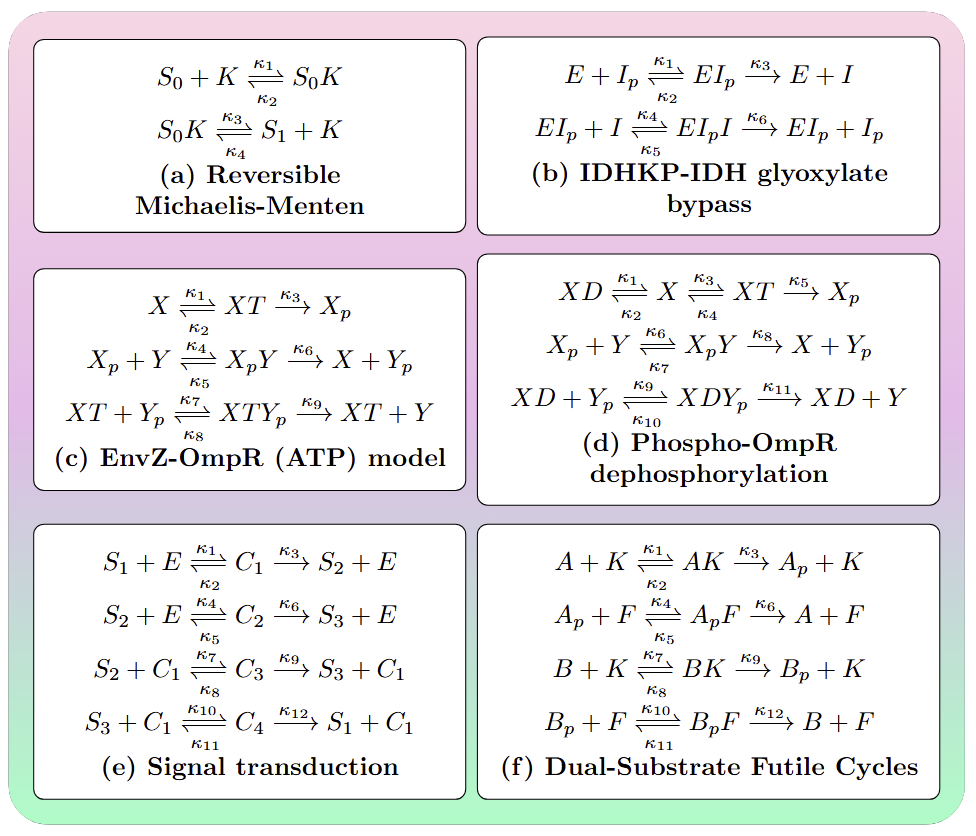}
    \caption{Biochemically motivated reaction networks from the literature that motivate the study of non‑redundant zero‑one networks.}
    \parbox{\linewidth}{\footnotesize Networks (a)--(f) are taken from: (a) \cite[Example 3.11]{Puente2025}; (b) \cite[Example 3.3]{Joshi2024}; (c) \cite[Fig.~2B]{Shinar2010}; (d) \cite[Fig.~2D]{Shinar2010}; (e) \cite[Example 4.17]{Puente2025}; (f) \cite[Supplement, \S 6.1]{Conradi2017a}.}
    \label{fig:networks_background}
\end{figure}

This paper gives a complete classification of non-vacuous ACR in non-redundant zero-one networks of stoichiometric dimension at most two, working directly from the stoichiometric matrix without explicit steady-state elimination. The classification is organized around a structural index $s^*$, defined as the number of distinct rows of the stoichiometric matrix.

Our main classification results are as follows.
\begin{enumerate}
    \item[(I)] \textbf{Dimension one.} The only non-redundant zero-one network with non-vacuous ACR for generic rate constants is the elementary inflow--outflow module $X_1\xrightleftharpoons[]{}0$ (Theorem~\ref{thm:one-dim ACR}).
    \item[(II)] \textbf{Dimension two, $s^*=2$.} A network admits non-vacuous ACR for generic rate constants if and only if it is nondegenerate and exactly one stoichiometric row is distinct from all remaining rows (Theorem~\ref{thm:2d,s,2s^*,has ACR i.f.f.}).
    \item[(III)] \textbf{Dimension two, $s^*=3$.} Non-vacuous ACR occurs if and only if the network is a natural species refinement (Definition~\ref{def:species_refinement}) of a consistent subnetwork of one of the five canonical prototypes in Fig.~\ref{fig:total_networks}, with the stoichiometric row of the ACR species appearing exactly once (Theorem~\ref{thm:2d,3s,3s^* with ACR table}).
    \item[(IV)] \textbf{Dimension two, $s^*\geq 4$.} No network admits non-vacuous ACR for any generic rate assignment (Theorem~\ref{thm:2d,s*≥4 no ACR}).
\end{enumerate}

\begin{figure}[htbp]
    \centering
    \begin{subfigure}{0.48\textwidth}
        \centering
        \includegraphics[height=1.9cm]{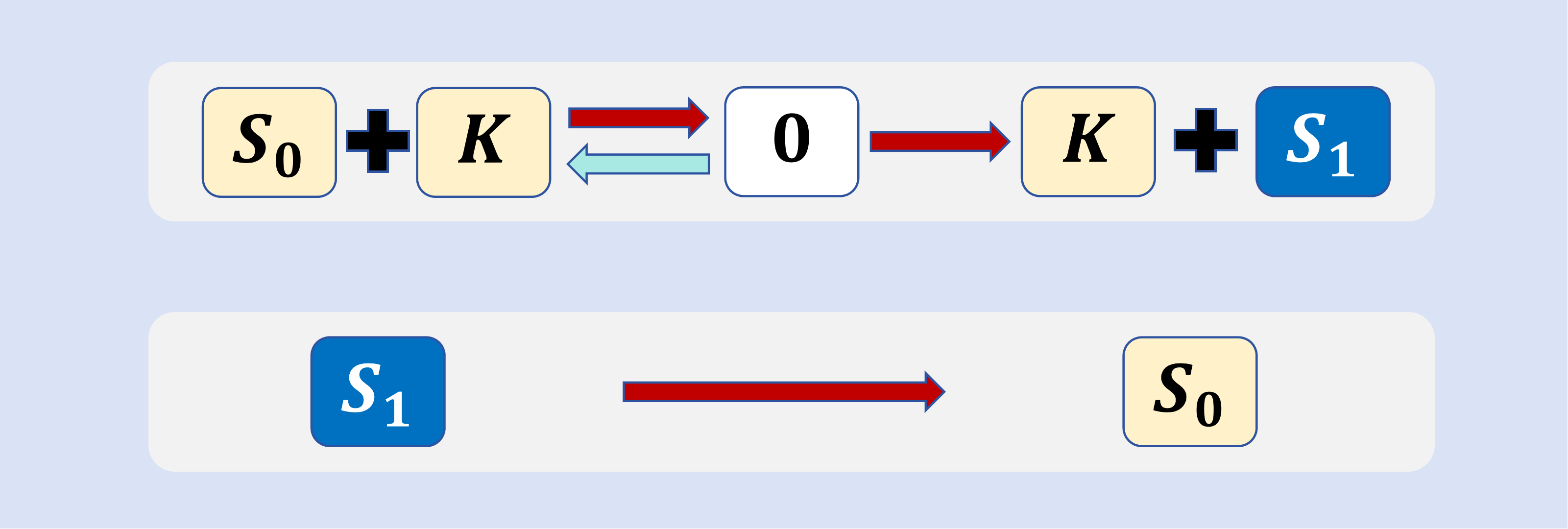}
        \caption{Conservative enzyme-catalyzed network}
        \label{fig:1a}
    \end{subfigure}
    \hfill
    \begin{subfigure}{0.48\textwidth}
        \centering
        \includegraphics[height=1.9cm]{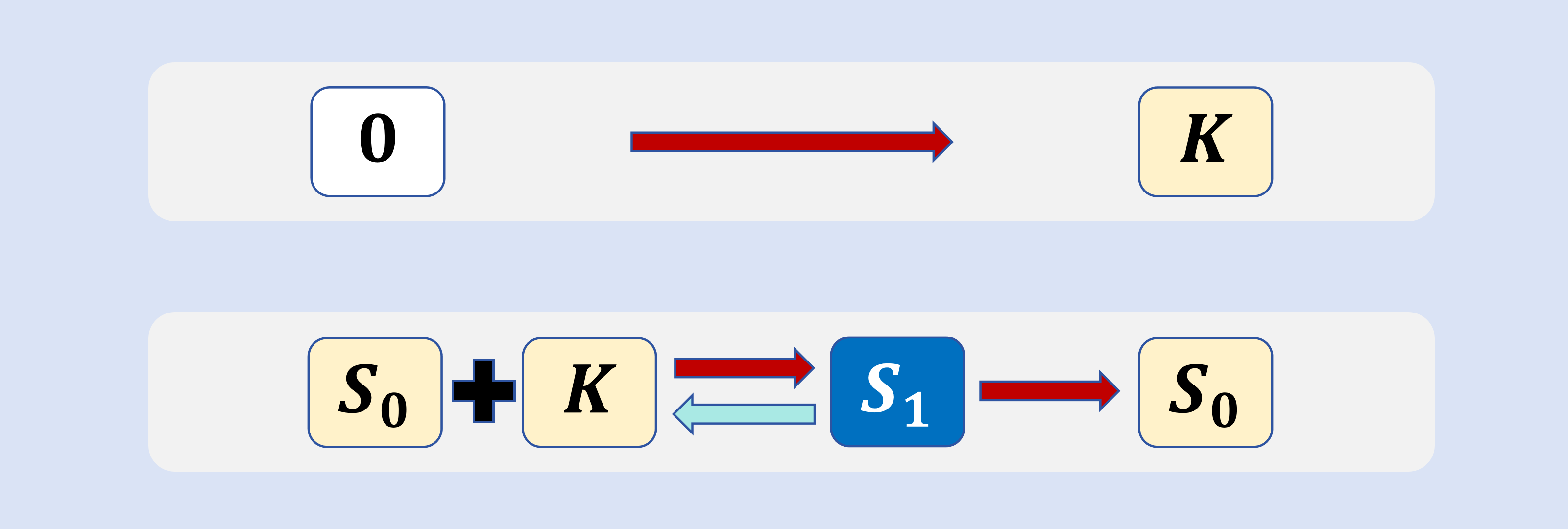}
        \caption{Anticoupled enzyme-catalyzed network}
        \label{fig:1b}
    \end{subfigure}
    \vspace{1pt}

    \begin{subfigure}{0.48\textwidth}
        \centering
        \includegraphics[height=2cm]{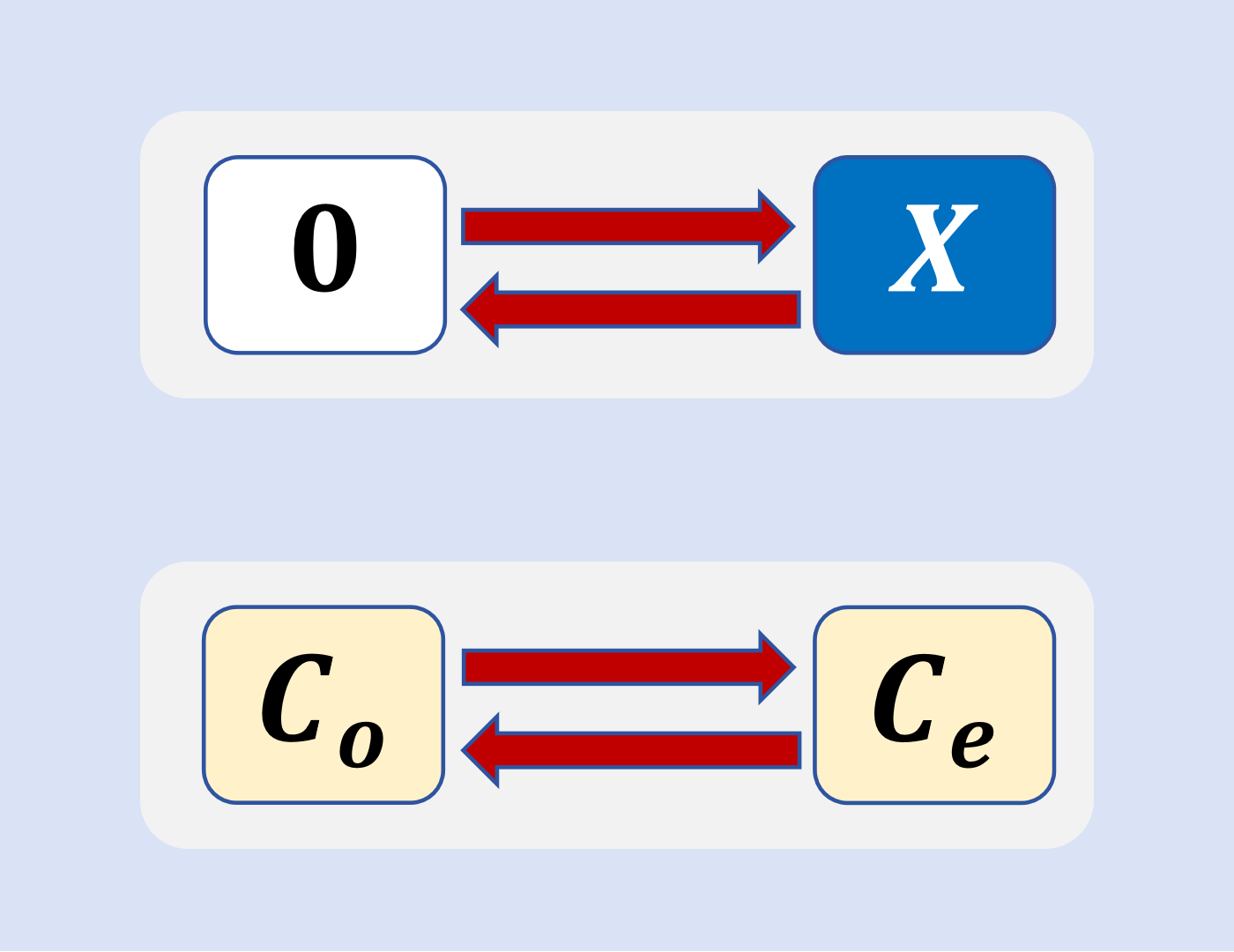}
        \caption{Decoupled carbon nanotube ropes}
        \label{fig:1c}
    \end{subfigure}
    \hfill
    \begin{subfigure}{0.48\textwidth}
        \centering
        \includegraphics[height=2cm]{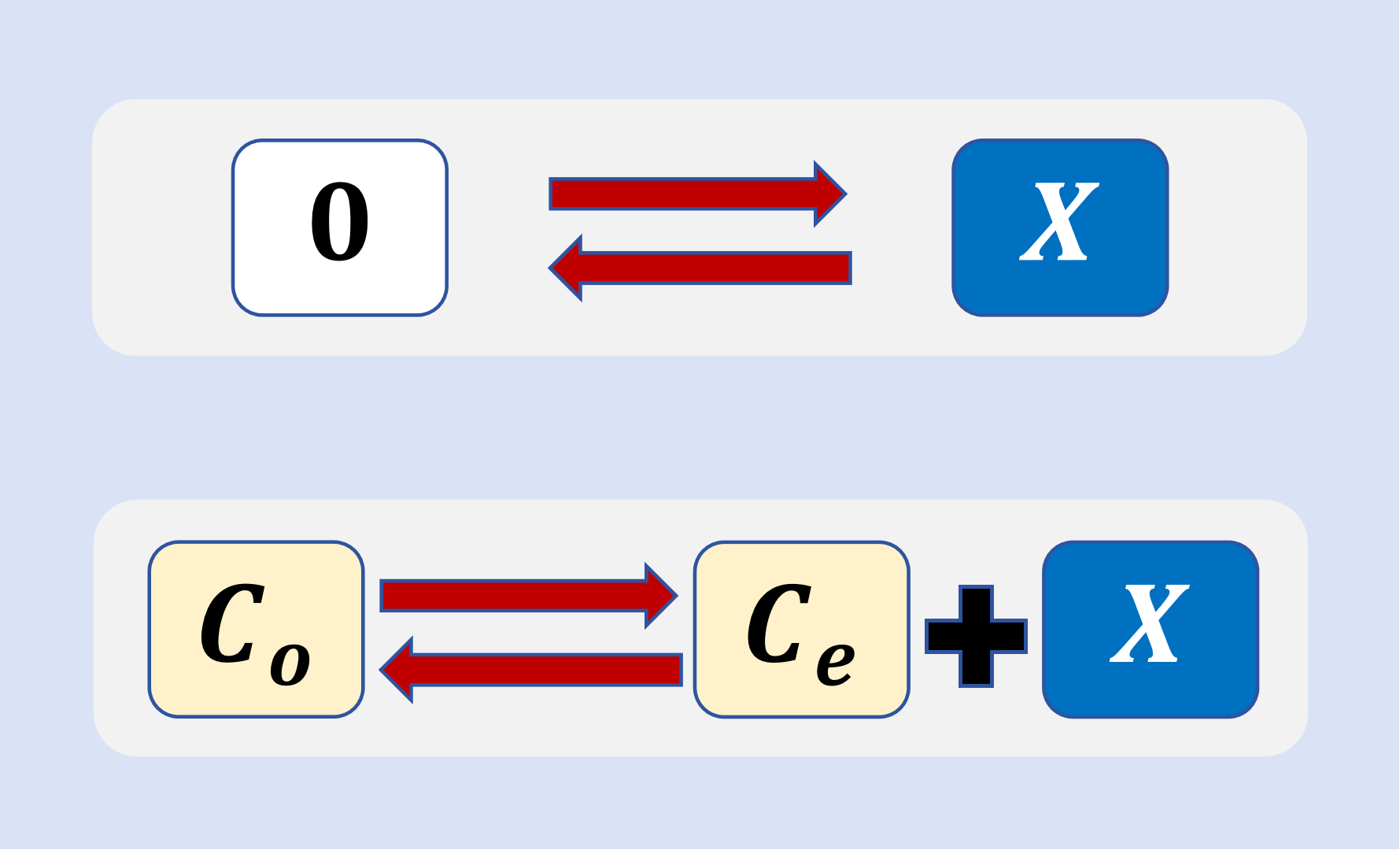}
        \caption{Inflow-outflow carbon nanotube ropes}
        \label{fig:1d}
    \end{subfigure}
    \vspace{1pt}

    \begin{subfigure}{0.95\textwidth}
        \centering
        \includegraphics[height=2cm]{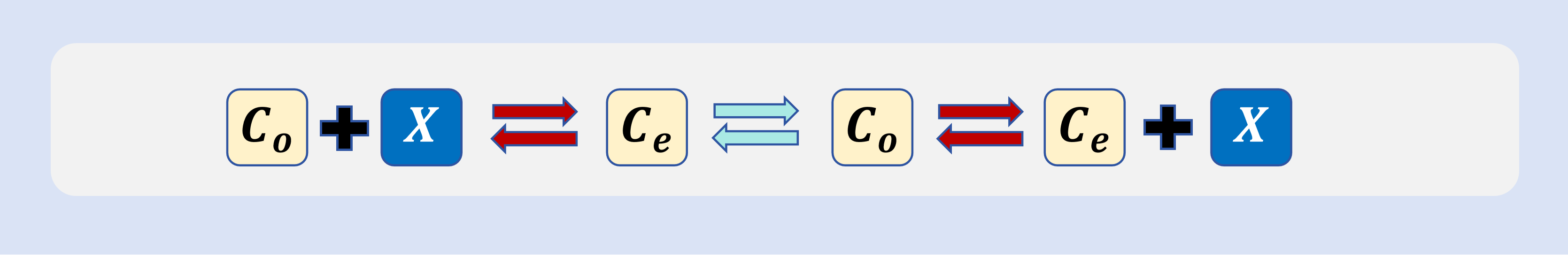}
        \caption{Connected carbon nanotube ropes}
        \label{fig:1e}
    \end{subfigure}
    \vspace{-8pt}
    \caption{Five canonical prototypes of two-dimensional ACR networks with $s^*=3$.}
    \label{fig:total_networks}  
    \vspace{-8pt}
\end{figure}

Result~(I) shows that one-dimensional ACR requires the absence of conservation laws. Results~(II)--(IV) reveal a sharp complexity threshold: as $s^*$ increases, the diversity of stoichiometric rows enriches conservation-law structure, increases the prevalence of nondegenerate networks obtained by species removal while preserving dimension, and systematically obstructs non-vacuous ACR. In the zero-one setting, ACR can persist only when this structural complexity remains sufficiently low, a condition precisely captured by the index $s^*$.

The five canonical prototypes in Fig.~\ref{fig:total_networks} connect this classification to classical models such as Michaelis--Menten kinetics \cite{MiMen1913} and carbon-nanotube ropes \cite{Cobden1998}, and complement the motif-based sufficient conditions of \cite{Joshi2023b}.

The paper is organized as follows.
Section~\ref{sec:backgrond} reviews preliminaries on reaction networks, ACR, and conservation laws. Section~\ref{sec:criterion} states Theorem~\ref{thm:main}, the mass-conserving reformulation of the necessary condition in \cite[Theorem 4.4]{Feliu2024}, which motivates and supports the classification. Section~\ref{sec:non-redundant network} presents the complete classification (Theorems~\ref{thm:one-dim ACR}--\ref{thm:2d,s*≥4 no ACR}). Section~\ref{sec:discussion} discusses extensions to higher dimensions.

\section{Preliminary}\label{sec:backgrond}
This section provides a concise overview of the standard terminology and definitions for reaction network. Comprehensive treatments can be found in \cite{Feinberg2019}. 

\subsection{Reaction network}

Formally, a \defword{reaction network} $G$ (often abbreviated to {\em network}) is defined by $s$ species $X_1$, $\cdots$, $X_s$ and $m$ reactions
\begin{equation}\label{eq:network}
\alpha_{j1}X_1 +
 \dots +
\alpha_{js}X_s
~ \xrightarrow{\kappa_j} ~
\beta_{j1}X_1 +
 \dots +
\beta_{js}X_s,\quad{\rm for}~ j=1,\;\ldots,\;m,
\end{equation}
where the \defword{stoichiometric coefficients} $\alpha_{ji}$ and $\beta_{ji}$ are non-negative integers. Without loss of generality, we assume that $(\alpha_{j1},\;\ldots,\;\alpha_{js})\neq(\beta_{j1},\;\ldots,\;\beta_{js})$. 
Each $\kappa_j\in \mathbb{R}_{>0}$ associated with the $j$-th reaction in \eqref{eq:network} is called a \defword{rate constant}. 
A network is \defword{zero-one} if all its stoichiometric coefficients $\alpha_{ji}$ and $\beta_{ji}$ take values in $\{0,1\}$. 

For a given network $G$, two fundamental matrices are defined as follows. First, the $s\times m$ matrix with $(i, j)$-entry equal to $\alpha_{ji}$ is the \defword{reactant matrix}, denoted by $\mathcal{Y}$. Second, the $s \times m$ matrix with $(i, j)$-entry equal to $\beta_{ji}-\alpha_{ji}$ is the \defword{stoichiometric matrix}, denoted by $\mathcal{N}$. 
The real linear subspace spanned by the columns of $\mathcal{N}$ is called the \defword{stoichiometric subspace}, denoted by $S$.
Notice that the dimension of $S$ equals ${\rm rank}(\mathcal N)$, and the \defword{dimension} of the network is defined as the dimension of  $S$. 
Given a network $G$, a \defword{subnetwork} of $G$ is a network obtained by removing some reactions from $G$ while preserving the same dimension as $G$. 

We say that \defword{two  matrices are equivalent} if one of the matrices can be obtained from the other by permuting the rows or the columns. Notice that the stoichiometric matrix $\mathcal{N}$ of a non-redundant zero-one network uniquely determines the reactant matrix $\mathcal{Y}$. We say that \defword{two non-redundant zero-one networks are equivalent} if their stoichiometric matrices are equivalent. 

\subsection{Steady state}
For each species $X_i$, its concentration is denoted by $x_i$. 
Under the assumption of {\em mass-action kinetics}, the time evolution of these concentrations is governed by the following ordinary differential equations (ODEs). 
\begin{equation}\label{eq:sys}
\dot{x} = f(\kappa, x) := \mathcal{N}\cdot\begin{pmatrix}
    \kappa_1 \prod_{i=1}^s x_i^{\alpha_{1i}}\\\vdots\\\kappa_m \prod_{i=1}^s x_i^{\alpha_{mi}}
\end{pmatrix},
\end{equation}
where $\kappa$ denotes the vector $(\kappa_1,\;\dots,\;\kappa_m)$, and $x$ denotes the vector $(x_1,\;\ldots,\;x_s)$.  
We denote by $Jac_f$  the Jacobian matrix of $f$ w.r.t. (with respect to) $x$. 

For a given rate-constant vector $\kappa^*\in \mathbb{R}_{>0}^m$, a \defword{steady state} of~\eqref{eq:sys} is a concentration vector $x^*\in\mathbb{R}_{>0}^s$ at which $f(\kappa, \;x)$ on the right-hand side of the ODEs \eqref{eq:sys} vanishes, i.e., $f(\kappa^*,\;x^*)=0$. We call the system \eqref{eq:sys} \defword{the steady-state system} of the network \eqref{eq:network}. 
We remark that all steady states in our context are strictly positive.  We say a steady state \( x^* \) is \defword{degenerate} if $
\text{im}\left(Jac_f(\kappa^*, x^*)|_S \right) \neq S.$
We say a network $G$ \defword{admits $n$ steady states} if there exists $\kappa^*\in \mathbb{R}_{>0}^m$ such that the network has $n$ steady states.  
If a network $G$ admits at least one steady state, then we say the network is \defword{consistent}. 
If a network $G$ admits at least two steady states, then we say the network  \defword{admits multistationarity}. 
If a network $G$ admits only degenerate steady states, then we say $G$ is \defword{degenerate}. Otherwise, if a network $G$ admits at least one nondegenerate steady state, then we say $G$ is \defword{nondegenerate}.


\subsection{ACR}\label{subsec:def of ACR}
For a given network $G$ and a fixed rate-constant vector $\kappa^*$, we say $G$ \defword{has non-vacuous ACR in species $X_i$} if $G$ has steady states and the value of $x_i^*$ is identical for every steady state $x^*$. We say a network has \defword{non-vacuous ACR} if it has non-vacuous ACR in some species. We remark that a network has \defword{vacuous ACR} if it is not consistent. 

\subsection{Conservation law}
Given a network $G$ as in \eqref{eq:network}, assume that $s > \mathrm{rank}(\mathcal{N})$. Let $d = s - \mathrm{rank}(\mathcal{N})$. A \defword{conservation-law matrix} of $G$, denoted by $W$, is any row-reduced $d \times s$ matrix whose rows constitute a basis for the orthogonal complement $S^{\perp}$, wherein $\mathrm{rank}(W) = d$. It follows that system \eqref{eq:sys} satisfies the identity $W\dot{x} = 0$. Consequently, any trajectory $x(t)$ originating from a non-negative vector $x(0)\in\mathbb{R}_{\geq0}^s$ stays, at all positive times, within the following \defword{stoichiometric compatibility class} associated with the \defword{total-constant vector} $c = Wx(0)\in\mathbb{R}^d$:
$\mathcal{P}_c = \bigl\{ x\in\mathbb{R}_{\geq0}^s \,\big|\, Wx = c \bigr\}$. 
Let $I = \{i_1,\dots,i_d\}$ denote the indices corresponding to the first non-zero entry in each row of $W$, and we assume that $i_1 < \cdots < i_d$. We then define $h_j$ as
\begin{equation}\label{eq:h_eta}
    h_j =
    \begin{cases}
        f_j & \text{if } j\notin I,\\
        (Wx - c)_k & \text{if } j = i_k \in I,
    \end{cases}
\end{equation}
where $f_1,\;\ldots,\;f_s$ are the polynomial functions specified in \eqref{eq:sys}. Thus, we define the system $h(\kappa,\;c,\;x)$ (abbreviated as $h$) as $h\equals(h_1,\;\ldots,\;h_s)$, and we call the system $h$ \defword{the steady-state system augmented by conservation laws}. Notice that a steady state $x^*$ is nondegenerate if and only if the Jacobian matrix $Jac_h(\kappa^*, x^*)$  has full rank \cite{Banaji2018}.

\subsection{Transformed Jacobian matrix}
For a network $G$ \eqref{eq:network}, let $\mathcal{N}$ and $\mathcal{Y}$ denote its stoichiometric matrix and reactant matrix, respectively. 
For the matrix $\mathcal{N}$, the corresponding \defword{flux cone} is defined by $\mathcal{F}(\mathcal{N}) := \left\{\mu\in \mathbb{R}_{\geq 0}^m \mid \mathcal{N} \mu = \mathbf{0} \right\}$. 
Let $\ell^{(1)},\ldots,\ell^{(n)}\in\mathbb{R}_{\geq0}^m$ be a set of generators of $\mathcal{F}(\mathcal{N})$. Then, any $\mu\in \mathcal{F}(\mathcal{N})$ can be written as $\mu=\sum_{i=1}^{n} \lambda_i \ell^{(i)}$, with $\lambda_i\geq0$ for all $i\in\{1,\;\ldots,\;n\}$. 
We define the \defword{partial transformed Jacobian matrix} in terms of $\lambda$ as \begin{equation}\label{eq:A}
    A(\lambda)\equals\mathcal{N}\;diag\left(\sum_{i=1}^n\lambda_i\ell^{(i)}\right)\;\mathcal{Y}^{\top}.
\end{equation}
Notice that the Jacobian matrix $Jac_f(\kappa,\;x)$ associated with the steady-state system $f$ \eqref{eq:sys} can be written as the \defword{transformed Jacobian matrix}
$J(p,\;\lambda) := A(\lambda)~diag(p)$, where $p=\left( p_1,\;\ldots,\;p_s \right)\equals\left(\begin{smallmatrix}
    \frac{1}{x_1},\;\ldots,\;\frac{1}{x_s} 
\end{smallmatrix}\right)$. 

Given an $k\times k$ matrix $M$, suppose $I,\;J\subseteq\{1,\;\ldots,\;k\}$ and $|I|=|J|$. The submatrix of $M$ comprising the rows indexed by $I$ and the columns indexed by $J$ is denoted by $M[I,\;J]$.

\begin{lemma}\cite{Feliu2024}\label{lm:degenerate equivals to A is zero poly}
    Suppose $G$ is an $r$-dimensional network with $s$ species. Assume that $A(\lambda)$ is the partial transformed Jacobian matrix of $G$. Then, $G$ is degenerate if and only if, for all $I\subseteq \{1,\;\ldots,\;s\}$ with $|I|=r$, $det(A(\lambda)[I,\;I])$ is the zero polynomial. 
\end{lemma}

\section{Rank-based criterion}\label{sec:criterion}

In this section, we provide  an obstruction principle for an $r$-dimensional $s$-species ($r<s$) network to admit no non-vacuous ACR in species $X_i$, see Theorem \ref{thm:main}, which is derived  from Lemma \ref{lm:degenerate equivals to A is zero poly} and  \cite[Theorem 4.4]{Feliu2024}. As an immediate corollary of Theorem \ref{thm:main}, Corollary \ref{coro: all s-1 nondegenerate then s no ACR} provides a criterion for detecting whether a network admits no non-vacuous ACR. Its usage is illustrated in Example \ref{exp: ACR‑2d 4s}. 
Although the converse of Theorem \ref{thm:main} does not hold in general (Example \ref{exp:reverse main}),  Remark \ref{rem:s-1 degenerate, x_i has ACR} indicates that the converse becomes valid when a network admits no multistationarity.
The  practical effectiveness and computational complexity  of the method are explained in Remark \ref{rem:thm_efficiency}.

\begin{theorem}[Obstruction Principle]\label{thm:main}
 Suppose $G$ is an $r$-dimensional network with $s$ species ($r<s$). Let $G_{/i}$ denote an $r$-dimensional network obtained by removing a single species $X_i$ from $G$. If $G_{/i}$ is nondegenerate, then for any generic rate-constant vector $\kappa \in {\mathbb R}^m_{>0}$, the network $G$ has no non-vacuous ACR in species $X_i$. 
\end{theorem}
\begin{proof}
Since $G_{/i}$ is obtained by removing the species $X_i$ from $G$ and $G_{/i}$ is nondegenerate, by 
Lemma \ref{lm:degenerate equivals to A is zero poly},
there exists a non-zero $r\times r$ principle minor of $A(\lambda)$, where 
$A(\lambda)$ is the partial transformed Jacobian matrix of $G$ defined as in \eqref{eq:A}. Let $\tilde A(\lambda)$ denote the matrix obtained by removing the $i$-th column of $A(\lambda)$. Then, 
the rank of $\tilde A(\lambda)$ is still $r$ for some choice of $\lambda$. So,
    by \cite[Theorem 4.4]{Feliu2024}, for any generic rate-constant vector $\kappa \in {\mathbb R}^m_{>0}$, the network $G$ has no non-vacuous ACR in species $X_i$. 
\end{proof}

\begin{remark}
In Theorem \ref{thm:main}, by a generic rate-constant vector $\kappa\in {\mathbb R}_{>0}^m$, we mean there exists an affine variety ${\mathcal V}\subset {\mathbb C}^m$, such that $\kappa \in {\mathbb R}_{>0}^m\backslash {\mathcal V}$, where an \defword{affine variety} is 
the set of common zeros in $\mathbb{C}^m$ of a family of polynomials.   
We provide a self-contained proof of Theorem \ref{thm:main} that does not depend on \cite[Theorem 4.4]{Feliu2024}, see Appendix \ref{app:proof}, where we  explicitly construct an affine variety ${\mathcal V}$ such that for any $\kappa \in {\mathbb R}_{>0}^m\backslash {\mathcal V}$, a given network $G$ has no non-vacuous ACR in $X_i$ when $G_{/i}$ is  nondegenerate.    
\end{remark}

\begin{corollary}\label{coro: all s-1 nondegenerate then s no ACR}
     Suppose $G$ is an $r$-dimensional network with $s$ species ($r<s$). For any network $G_{/i}$ obtained by removing a single species $X_i$ from $G$, if $G_{/i}$ is $r$-dimensional and nondegenerate, then for any generic rate-constant vector $\kappa\in\mathbb{R}^m_{>0}$, the network $G$ has no non-vacuous ACR. 
\end{corollary}

\begin{example}\cite[Example 14]{Joshi2021}(Nanotube rope network)
\label{exp: ACR‑2d 4s}
\begin{equation*}
C_o + X\xrightleftharpoons[\kappa_2]{\kappa_1} C_e\;\;,\quad
C_o\xrightleftharpoons[\kappa_4]{\kappa_3} C_e + X.
\end{equation*}
Let $X_1$, $X_2$, and $X_3$ denote the species $C_o$, $C_e$, and $X$, respectively. 

Notice that
$\mathcal{N}=\left(\begin{smallmatrix}
-1&\;\;1&-1&\;\;1\\
\;\;1&-1&\;\;1&-1\\
-1&\;\;1&\;\;1&-1
\end{smallmatrix}\right)$, and
$\mathcal{Y}=\left(\begin{smallmatrix}
\;\;1&\;\;0&\;\;1&\;\;0\\
\;\;0&\;\;1&\;\;0&\;\;1\\
\;\;1&\;\;0&\;\;0&\;\;1
\end{smallmatrix}\right)$.
Also note that $r=\operatorname{rank}(\mathcal{N})=2$, and $\mathcal{N}_2$, $\mathcal{N}_3$ are linearly independent. The vectors $\ell^{(1)}=(0,0,1,1)$, and $\ell^{(2)}=(1,1,0,0)$ are the extreme rays of the flux cone $\mathcal{F}(\mathcal{N})$. Hence, by \eqref{eq:A}, we have
\begin{equation*}
A(\lambda)=\left(\begin{smallmatrix}
-\lambda_2-\lambda_1&\;\;\lambda_2+\lambda_1&-\lambda_2+\lambda_1\\
\;\;\lambda_2+\lambda_1&-\lambda_2-\lambda_1&\;\;\lambda_2-\lambda_1\\
-\lambda_2+\lambda_1&\;\;\lambda_2-\lambda_1&-\lambda_2-\lambda_1
\end{smallmatrix}\right).
\end{equation*}
For $i=1$, let $\mathcal{N}_{/1}$ denote the stoichiometric matrix corresponding to the network obtained by removing species $X_1$ from the given network. Notice that $\operatorname{rank}\left(\mathcal{N}_{/1}\right)=r$, and $\det(A(\lambda)[\{2,3\},\;\{2,3\}])=4\lambda_1\lambda_2$. 
For $i=2$, in the same way, we remove species $X_2$ from the given network. Notice that $\operatorname{rank}\left(\mathcal{N}_{/2}\right)=r$, and $\det(A(\lambda)[\{1,3\},\;\{1,3\}])=4\lambda_1\lambda_2$. 
By Lemma \ref{lm:degenerate equivals to A is zero poly}, the $2$-species  networks  $G_{/1}$ and $G_{/2}$ are nondegenerate. Hence, by Theorem  \ref{coro: all s-1 nondegenerate then s no ACR}, the given network has no non‑vacuous ACR in species $X_1$ or $X_2$ for any generic choice of rate constants.  
For $i=3$, in the same way, we remove species $X_3$ from the given network. Notice that $\operatorname{rank}(\mathcal{N}_{/3})=r-1$ and 
Theorem \ref{thm:main} is not directly  applicable.  However, by a similar idea to the proof of Theorem \ref{thm:main} (Appendix \ref{app:proof}), we  find  that $\det(A(\lambda)[\{2,3\},\;\{1,2\}])$ is identically zero, and we conclude that the network possibly admits non-vacuous ACR in $X_3$.  
\end{example}

\begin{example}\label{exp:reverse main}
    This example demonstrates that Theorem \ref{thm:main} provides a necessary but not  sufficient condition for the existence of non-vacuous ACR. Consider the following network $G$: $3X_1+X_2\xrightarrow{\kappa_1}4X_1,\;\;2X_1+X_2\xrightarrow{\kappa_2}3X_2,\;\;X_1+X_2\xrightarrow{\kappa_3}2X_1$. The network $G_{/1}$ obtained by removing species $X_1$ from $G$ is $X_2\xrightarrow{\kappa_1}0,\;X_2\xrightarrow{\kappa_2}3X_2,\;X_2\xrightarrow{\kappa_3}0$. 
    The steady-state system of $G_{/1}$ is $\tilde{f}=x_2(-\kappa_1+2\kappa_2-\kappa_3)$. Notice that if $-\kappa_1+2\kappa_2-\kappa_3\neq0$, then $G_{/1}$ has no positive steady states. If $-\kappa_1+2\kappa_2-\kappa_3\equiv0$, then $G_{/1}$ has infinitely many degenerate positive steady states for specific rate constants, and $G_{/1}$ is degenerate. Notice that $G_{/1}$ has the same dimension as $G$. 
    
    The steady-state system augmented by the conservation law of $G$ is 
    \begin{equation*}\begin{cases}
        h_1&= x_1 x_2(\kappa_1x_1^2-2\kappa_2x_1+\kappa_3),\\
    h_2&= x_1+x_2-c. 
    \end{cases}\end{equation*}
    Let $h_1=0$ and $h_2=0$. For any generic choice of rate constants, the network $G$ either has two positive steady states with distinct coordinate values for $x_1$, or has no positive steady states. Hence, $G$ admits no non-vacuous ACR in species $X_1$ generically. 
\end{example}

\begin{remark}\label{rem:s-1 degenerate, x_i has ACR}
    This remark presents the contrapositive statement of Theorem \ref{thm:main} under additional conditions, which follows from the proof of Theorem \ref{thm:main} \cite[Theorem 4.4]{Feliu2024}.  Let $G_{/i}$ be a network obtained by removing a single species $X_i$ from $G$, where $G_{/i}$ remains the same dimension as $G$. If $G_{/i}$ is degenerate and $G$ admits exactly one positive steady state, then $G$ has non-vacuous ACR in the removed species $X_i$ for any generic choice of rate constants.  Since the coexistence of multistationarity and ACR is known to be rare with high probability \cite{JoKaAn2023}, Theorem \ref{thm:main} is highly effective in applications, see Appendix \ref{subsec:alg}--Table~\ref{tab:algorithm summary}. 
\end{remark}

\begin{example}
We provide an example for Remark \ref{rem:s-1 degenerate, x_i has ACR}.
    Consider the following network $G$: $X_3\xrightarrow{\kappa_1}X_1,\;\; 0\xrightarrow{\kappa_2}X_2+X_3,\;\;X_1+X_2\xrightarrow{\kappa_3}0$. 
    The network $G_{/3}$ obtained by removing species $X_3$ from $G$ is $0\xrightarrow{\kappa_1}X_1,\;\;0\xrightarrow{\kappa_2}X_2,\;\;X_1+X_2\xrightarrow{\kappa_3}0$. The steady-state system $\tilde{f}$ of $G_{/3}$ is composed of $\tilde{f}_1= \kappa_1 - \kappa_3x_1x_2$, and $\tilde{f}_2= \kappa_2 - \kappa_3x_1x_2$. Notice that $G_{/3}$ has the same dimension as $G$.   Since $G_{/3}$ has no conservation laws, $\tilde{h}=\tilde{f}$. Then, we have that $det(Jac_{\tilde{h}})=\left|\begin{smallmatrix}
        -\kappa_3x_2 & -\kappa_3x_1 \\
            -\kappa_3x_2 & -\kappa_3x_1 
    \end{smallmatrix}\right|$ is the zero polynomial at all positive steady states. Hence, $G_{/3}$ is degenerate. 
    
   One may notice by \cite[Theorem 3]{Jiao2025} that $G$ admits no  multistationarity since it is a nondegenerate $2$D $3$-species zero-one network. 
 On the other hand, for the network $G$, note that $f_3=-\kappa_1x_3+\kappa_2$, where $f$ is the steady-state system of $G$. By $f_3=0$, for any generic choice of rate constants, we can see that $G$ has non-vacuous ACR in species $X_3$ since $x_3=\frac{\kappa_2}{\kappa_1}$. 
\end{example}

\begin{remark}\label{rem:thm_efficiency}
For Example~\ref{exp: ACR‑2d 4s},  we can indeed verify the non‑vacuous ACR in species $X_3$ as follows. The steady‑state system is
\begin{equation*}
\begin{cases}
f_1&=-\kappa_1x_1x_3+\kappa_2x_2-\kappa_3x_1+\kappa_4x_2x_3,\\
f_2&=\;\;\;\kappa_1x_1x_3-\kappa_2x_2+\kappa_3x_1-\kappa_4x_2x_3,\\
f_3&=-\kappa_1x_1x_3+\kappa_2x_2+\kappa_3x_1-\kappa_4x_2x_3.
\end{cases}
\end{equation*}
By $f_2+f_3=0$, we have $x_1 = \frac{\kappa_4x_2x_3}{\kappa_3}$. By $f_2-f_3=0$, we have $x_1 = \frac{\kappa_2x_2}{\kappa_1x_3}$. Hence, we obtain $\frac{\kappa_4x_2x_3}{\kappa_3} = \frac{\kappa_2x_2}{\kappa_1x_3}$, i.e., $x_3=\sqrt{\frac{\kappa_2\kappa_3}{\kappa_1\kappa_4}}$, which verifies the existence of non-vacuous ACR. For small networks, one can similarly confirm the existence of non-vacuous ACR by solving the steady‑state system algebraically.
However, for larger networks like those shown in Fig.~\ref{fig:networks_background},  it might be expensive or technical to solve the systems~\cite{Shinar2010, Joshi2023a}. By Remark~\ref{rem:s-1 degenerate, x_i has ACR}, one way is to check if these systems admit no multistationarity by looking at the deficiency~\cite{Puente2025}. Systematically, we develop Algorithm~\ref{alg:main}, which takes only seconds to provide reliable ACR information for the list of biochemical networks (Appendix \ref{subsec:alg}). In practical implementations, the method illustrated in Algorithm~\ref{alg:main} shows that detecting ACR requires computing at most $s\cdot\binom{s}{r}$ determinants of $r\times r$ matrices, polynomial in network size, in contrast to the exponential complexity of Gr\"{o}bner basis elimination.
\end{remark}

\section{Non-redundant zero-one networks}\label{sec:non-redundant network}
 We discuss $1$D and $2$D non-redundant zero-one networks in Section \ref{subsec: one-dimensional} and  Section \ref{subsec: two-dimensional}, respectively.  

\begin{definition}\label{def:non-redundant species & network}
    A zero-one network is called \defword{non-redundant} if there is no species appearing on both sides of any reaction.  
\end{definition}

\subsection{$\mathbf{1}$D networks}\label{subsec: one-dimensional}

\begin{theorem}\label{thm:one-dim ACR}
    Suppose $G$ is a $1$D non-redundant zero-one network. Then, $G$ has non-vacuous ACR for generic rate constants if and only if $G$ is equivalent to \begin{equation}\label{al: net of one-dim and one species}
        X_1\xrightarrow{\kappa_1}0,\quad0\xrightarrow{\kappa_2}X_1.
    \end{equation}
\end{theorem}
\begin{proof}
   ($\Leftarrow$)   It is straightforward to observe that $G$ has non-vacuous ACR in species $X_1$ for any generic choice of rate constants since  
the steady-state system of \eqref{al: net of one-dim and one species} is 
    $f=-\kappa_1x_1+\kappa_2$.  

    ($\Rightarrow$)By \cite[Lemma 14]{Jiao2025}, any $1$D zero-one network is nondegenerate. So, by Theorem \ref{thm:main}, if  $G$ admits non-vacuous ACR generically, then $G$ contains only one species. Thus, the conclusion follows from the fact that a $1$-species non-redundant zero-one network must be equivalent to \eqref{al: net of one-dim and one species}. 
\end{proof}

The above result implies that, among all $1$D networks, the only network that admits non-vacuous ACR is the one with no conservation laws. Hence, we infer that conservation laws may preclude the emergence of non-vacuous ACR. In practical applications, multi-species networks are more common, and such networks usually contain many conservation laws. We therefore proceed to discuss $2$D networks with multiple species, and explore whether networks with numerous conservation laws  admit non-vacuous ACR or not. 

\subsection{$\mathbf {2}$D networks}\label{subsec: two-dimensional}

Suppose $G$ is a non-redundant zero-one network with a rank-$2$ stoichiometric matrix $\mathcal{N}$. Let $\mathcal{N}_k$ denote the $k$-th row of $\mathcal{N}$. Since ${\rm rank}(\mathcal{N})=2$, any row of $\mathcal{N}$ can be expressed as a linear combination of two linearly independent rows of $\mathcal{N}$. 
We present a key notation for Section \ref{subsec: two-dimensional}. Let $s^*$ denote the number of distinct rows in the stoichiometric matrix $\mathcal{N}$. 
Since ${\rm rank}(\mathcal{N})=2$, we have $s^*\geq2$. Then, we discuss the $2$D networks in Section \ref{sssec:2d, s^*=2}, Section \ref{sssec:2d,s^*=3}, and Section \ref{sssec:2d,s^*>=4}, corresponding to the cases $s^*=2$, $s^*=3$, and $s^*\geq4$, respectively.

\subsubsection{Networks with $\mathbf{s^*=2}$}\label{sssec:2d, s^*=2}
In this section, we state the main result in Theorem \ref{thm:2d,s,2s^*,has ACR i.f.f.}, where for networks with $s^*=2$, we characterize all networks that admit non-vacuous ACR. We first classify all networks with $s^*=2$ into two classes according to whether there exists a row of $\mathcal{N}$ that appears only once, as detailed in cases (i) and (ii) of Remark \ref{rem:2d,2s,2s*,classification}. By Lemma \ref{lm: symmetry for no ACR}, we show that networks in case (i) admit no non-vacuous ACR. 
To complete the discussion of  case (ii), we characterize all $2$-species degenerate networks in Lemma \ref{lm:2d,2s,degnerate if and only if 1378=0},  
and the final proof implies that any $s$-species ACR network should be built based on a nondegenerate $2$-species network. 

\begin{theorem}\label{thm:2d,s,2s^*,has ACR i.f.f.}
    Suppose $G$ is a $2$D non-redundant zero-one network and $s^*=2$. Assume that the first two rows of the stoichiometric matrix $\mathcal{N}$ ($\mathcal{N}_1$, $\mathcal{N}_2$) are linearly independent. Then, $G$ has non-vacuous ACR in species $X_i$ for any generic choice of rate constants if and only if both of the following conditions hold. \begin{enumerate}
        \item [(i)] {$\begin{pmatrix}
        \mathcal{N}_1\\\mathcal{N}_2
    \end{pmatrix}$ is not equivalent to $\left(\begin{array}{rrrr}
        \;\;1&\;\;0&-1&\;\;1\\0&1&-1&1
    \end{array}\right)$ or $\left(\begin{array}{rrr}
        \;\;1&\;\;0&-1\\0&1&-1
    \end{array}\right)$. }
        \item [(ii)] {$\mathcal{N}_k=\mathcal{N}_1$ for all $k\in\{3,\;\ldots,\;s\}\setminus\{i\}$, or $\mathcal{N}_k=\mathcal{N}_2$ for all $k\in\{3,\;\ldots,\;s\}\setminus\{i\}$. }
    \end{enumerate}
\end{theorem}

To complete the proof of Theorem \ref{thm:2d,s,2s^*,has ACR i.f.f.}, we first prepare some useful lemmas. 

\begin{remark}\label{rem:2d,2s,2s*,classification}
    Suppose $G$ is a network with $s$ species and $s^*=2$.  Assume that the first two rows of the stoichiometric matrix $\mathcal{N}$ are distinct. If $s\geq3$, then one of the following holds. 
    \begin{enumerate}
        \item [(i)] {There exist $i,\;j\in\{3,\;\ldots,\;s\}$ such that $\mathcal{N}_i=\mathcal{N}_1$ and $\mathcal{N}_j=\mathcal{N}_2$. }
        \item [(ii)] {For all $k\in\{3,\;\ldots,\;s\}$, $\mathcal{N}_k=\mathcal{N}_1$, or for all $k\in\{3,\;\ldots,\;s\}$, $\mathcal{N}_k=\mathcal{N}_2$. }
    \end{enumerate}
\end{remark}

\begin{lemma}\label{lm: symmetry for no ACR}
  If two rows in the stoichiometric matrix  $\mathcal{N}$  of a non-redundant zero-one network $G$ are the same, say $\mathcal{N}_i=\mathcal{N}_j$, then $G$ has non-vacuous ACR in neither species $X_i$ nor species $X_j$ for any choice of rate constants. 
\end{lemma}
\begin{proof}
  Notice that the new network obtained by relabeling $X_i$ and $X_j$ in $G$ is equivalent to $G$. Hence, if there exists a rate-constant vector $\kappa^*$ such that $G$ has non-vacuous ACR in $X_i$ (or in $X_j$), then 
  $G$ also has non-vacuous ACR in $X_j$ (or in $X_i$), which contradicts the conservation law $x_i=x_j+c$. 
\end{proof}

\begin{lemma}\label{lm:2d,2s,degnerate if and only if 1378=0}
    Suppose $G$ is a $2$D non-redundant zero-one network with $2$ species. Then, $G$ is degenerate if and only if  the stoichiometric matrix
    $\mathcal{N}$ is equivalent to \begin{equation}\label{al:stoichiometric of 2s degenerate}
        \left(\begin{array}{rrrr}
            \;\;1&\;\;0&-1&\;\;1\\0&1&-1&1
        \end{array}\right)\text{ or }\left(\begin{array}{rrr}
            \;\;1&\;\;0&-1\\0&1&-1
        \end{array}\right).\;
    \end{equation}
\end{lemma}
\begin{proof}
    Notice that for a $2$D non-redundant zero-one network with $2$ species, $\mathcal{N}$ must be composed of certain columns of \begin{equation}\label{ma:N8 of lm 2s degenerate}
        \left(\begin{array}{rrrrrrrr}
        -1&\;\;1&0&\;\;0&-1&\;\;1&-1&1\\0&0&-1&1&-1&1&1&-1
    \end{array}\right). 
    \end{equation}
     The steady-state system corresponding to \eqref{ma:N8 of lm 2s degenerate} is 
    \begin{equation}\label{sys:total_exm_of_lm_degenerate}
        \begin{cases}
            f_1&=-\kappa_1x_1+\kappa_2-\kappa_5x_1x_2+\kappa_6-\kappa_7x_1+\kappa_8x_2,\;\\f_2&=-\kappa_3x_2+\kappa_4-\kappa_5x_1x_2+\kappa_6+\kappa_7x_1-\kappa_8x_2,\;
        \end{cases}\;
    \end{equation}where $\kappa\in\mathbb{R}^8_{\geq0}$. 
    Notice that if the $i$-th column of \eqref{ma:N8 of lm 2s degenerate} does not appear in $G$, then $\kappa_i=0$ in \eqref{sys:total_exm_of_lm_degenerate}.
    Since $G$ contains no conservation laws, by \eqref{eq:h_eta}, we have $f=h$. Hence, by \eqref{sys:total_exm_of_lm_degenerate}, we have 
    \begin{equation}\label{al:2d,2s,det(Jach)}
        det(Jac_h)=\kappa_1\kappa_3+\kappa_1\kappa_8+\kappa_3\kappa_7+\kappa_5(\kappa_1+2\kappa_7)~x_1+\kappa_5(\kappa_3+2\kappa_8)~x_2.\;
    \end{equation}
    
    ($\Leftarrow$)Since $\kappa_1=\kappa_3=\kappa_7=\kappa_8=0$, by  \eqref{al:2d,2s,det(Jach)}, we have that $det(Jac_h)$ is the zero polynomial. Hence, $G$ is degenerate. 
    
    ($\Rightarrow$) If $G$ is degenerate, then $det(Jac_h)$ is the zero polynomial. By \eqref{al:2d,2s,det(Jach)}, one of the following cases holds: (i) $\kappa_1=\kappa_3=\kappa_5=0$, (ii) $\kappa_1=\kappa_5=\kappa_7=0$, (iii) $\kappa_3=\kappa_5=\kappa_8=0$, and (iv) $\kappa_1=\kappa_3=\kappa_7=\kappa_8=0$. 
    Notice that only networks satisfying case (iv) are consistent. So their stoichiometric matrices can only be equivalent to  \eqref{al:stoichiometric of 2s degenerate}. 
\end{proof}

\noindent\textbf{\textit{Proof of Theorem \ref{thm:2d,s,2s^*,has ACR i.f.f.}}}
   ($\Leftarrow$) Assume that statements (i) and (ii) of Theorem \ref{thm:2d,s,2s^*,has ACR i.f.f.} hold simultaneously.  Let $\widetilde{G}$ be the non-redundant zero-one network corresponding to the first two rows of ${\mathcal N}$. Since $\widetilde{G}$ is zero-one,  its steady-state system $\tilde{f}$ can be written as \begin{equation}\label{sys:2d,2s}
        \begin{cases}
            \tilde{f}_1&=x_1g_1(\kappa,\;x_2)+g_2(\kappa,\;x_2),\;\\
        \tilde{f}_2&=x_1g_3(\kappa,\;x_2)+g_4(\kappa,\;x_2).\;
        \end{cases}\;
    \end{equation}
By condition (ii), without loss of generality, assume that $\mathcal{N}_k=\mathcal{N}_1$ for all $k\in\{3,\;\ldots,\;s\}$. Let 
        $Y=\prod_{\{k|\mathcal{N}_k=\mathcal{N}_1\}}x_k$.
   Then, the steady-state system of $G$ can be written as     \begin{equation}\label{sys:2d,2s*,s,Nk=N1}
        f=\tilde{f}|_{x_1=Y}. 
    \end{equation} 
   Since condition (i)  holds, by Lemma \ref{lm:2d,2s,degnerate if and only if 1378=0}, we know that $\widetilde{G}$ is nondegenerate. By \cite[Theorem A]{Feliu2024} and \cite[Theorem 3]{Jiao2025},  for any generic choice of rate-constant vector $\kappa^*$,  $\widetilde{G}$ has only one nondegenerate positive steady state $(x_1^*,\;x_2^*)$. By 
   \eqref{sys:2d,2s}, $(\kappa^*, x_2^*)$ definitely satisfies the implicit function 
   $x_2^*=x_2(\kappa^*)$ defined by 
   \begin{equation}\notag\begin{vmatrix}
        g_1(\kappa,\;x_2) & g_2(\kappa,\;x_2) \\
        g_3(\kappa,\;x_2) & g_4(\kappa,\;x_2)
    \end{vmatrix}=0. 
    \end{equation}
    So, by \eqref{sys:2d,2s} and \eqref{sys:2d,2s*,s,Nk=N1}, 
    any steady state $x^*$ of $G$ also satisfies $x_2^*=x_2(\kappa^*)$. That means the network $G$ has non-vacuous ACR in species $X_2$ for any generic choice of rate constants.  

    ($\Rightarrow$)
     If $G$ has non-vacuous ACR for any generic choice of rate constants, then by Lemma \ref{lm: symmetry for no ACR} and Remark \ref{rem:2d,2s,2s*,classification}, we have either $\mathcal{N}_k=\mathcal{N}_1$ for all $k\in\{3,\;\ldots,\;s\}$ or $\mathcal{N}_k=\mathcal{N}_2$ for all $k\in\{3,\;\ldots,\;s\}$, which is exactly the  condition (ii).
   Assume that the condition (i) does not hold. If the first two rows of 
   ${\mathcal N}$ are $\left(\begin{smallmatrix}
       \;\;1&\;\;0&-1&\;\;1\\\;\;0&\;\;1&-1&\;\;1
   \end{smallmatrix}\right)$,   
then, under condition (ii),  the steady-state system is 
\begin{equation}\label{sys:2d,s,degenerate}
        \begin{cases}
            f_1=\kappa_1-\kappa_3\prod_{i=1}^s x_i+\kappa_4,\\
            f_2=\kappa_2-\kappa_3\prod_{i=1}^s x_i+\kappa_4.
        \end{cases}\;
    \end{equation}
    Notice that the network even has no positive steady states when $\kappa_1\neq \kappa_2$,  which 
    contradicts the hypothesis that $G$ has non-vacuous ACR for  generic rate constants.  
    If the first two rows of 
  ${\mathcal N}$ are $\left(\begin{smallmatrix}
       \;\;1&\;\;0&-1\\\;\;0&\;\;1&-1
   \end{smallmatrix}\right)$, then the contradiction happens in the same way by setting $\kappa_4=0$ in \eqref{sys:2d,s,degenerate}. \hfill$\square$

\subsubsection{Networks with $\mathbf{s^*=3}$}\label{sssec:2d,s^*=3}
In this section, we state the main results in Theorem \ref{thm:2d,3s,3s^* with ACR table}, where we characterize all networks with $s^*=3$ that admit non-vacuous ACR. We prove the main results in a constructive way in two steps. The first step  focuses on $3$-species networks with non-vacuous ACR. Notice that any $2$-species network obtained by removing a single species from a $3$-species network is either $1$D or $2$D. We characterize the two cases respectively in Lemma \ref{lm:if 2s degenerate, then 3-th species has ACR}  and Lemma \ref{lm:if 2s 1d, then 3-th species has ACR}. 
Then, we identify all the $3$-species ACR networks in Lemma \ref{lm:2d,3s,equival,has ACR in X1,X2,X3 i.f.f.}, which explains the structures of the stoichiometric matrices \eqref{mat:1}--\eqref{mat:5} in Theorem~\ref{thm:2d,3s,3s^* with ACR table}.   The second step is to prove that any $s$-species $(s\geq 3)$ network with non-vacuous ACR must be built based on a $3$-species ACR network in a structural way, which is given in Lemma \ref{lm:2d,s,3s^*}. 

\begin{definition}\label{def:species_refinement}
Let $G$ be a network defined as in \eqref{eq:network} with species $X_1,\;\ldots,\;X_s$. For each $i\in\{1,\;\ldots,\;s\}$, choose $n_i\geq 1$ new species $Z_{i,1},\;\ldots,\;Z_{i,n_i}$ so that no species appears in more than one of the lists $\{Z_{i,1},\;\ldots,\;Z_{i,n_i}\}$. Define a map $\phi$ on the species set $\{X_1,\;\ldots,\;X_s\}$ by
\begin{equation*}
\phi(X_i)=Z_{i,1}+\cdots+Z_{i,n_i},\quad i=1,\;\ldots,\;s.
\end{equation*}
We call $\phi$ a \defword{natural species refinement map}. In every reaction of $G$, replace each $X_i$ by $\phi(X_i)$, keeping the same stoichiometric coefficients and rate constants. Denote the network obtained in this way by $\phi(G)$:
\begin{equation*}
\sum_{i=1}^s \alpha_{ij}\,\phi(X_i) \xrightarrow{\kappa_j} \sum_{i=1}^s \beta_{ij}\,\phi(X_i),\quad j=1,\;\ldots,\;m.
\end{equation*}
We call $\phi(G)$ a \defword{natural species refinement} of $G$, and $G$ a \defword{lumping network} of $\phi(G)$.
\end{definition}

\begin{theorem}\label{thm:2d,3s,3s^* with ACR table}
    Suppose $G$ is a $2$D non-redundant zero-one network with $s$ species and $s^*=3$. 
    Then, $G$ has non-vacuous ACR in species $X_i$ for any generic choice of rate constants if and only if there exist a $2$D non-redundant zero-one network $\widetilde{G}$ with $3$ species $\widetilde{X}_1$, $\widetilde{X}_2$, $\widetilde{X}_3$, a consistent subnetwork $\widetilde{G}_{\mathrm{sub}}$ of $\widetilde{G}$, and a natural species refinement map $\phi$ such that $G=\phi(\widetilde{G}_{\mathrm{sub}})$ and $\phi(\widetilde{X}_3)=X_i$, where the stoichiometric matrix $\widetilde{\mathcal{N}}$ of $\widetilde{G}$ is one of \eqref{mat:1}--\eqref{mat:5}.
    \par\smallskip
    \noindent
    \begin{minipage}{0.48\textwidth}
    \begin{equation}\label{mat:1} 
        \left(\begin{array}{rrrr}
        \cellcolor{lightblue!120} 1 & \cellcolor{lightblue!120} \;\;0 & \cellcolor{lightblue!120} -1 & \cellcolor{lightblue!120} \;\;1 \\
        \cellcolor{lightblue!120} 0 & \cellcolor{lightblue!120} 1 & \cellcolor{lightblue!120} -1 & \cellcolor{lightblue!120} 1 \\
        -1 & 1 & 0 & 0
    \end{array}\right)\quad
    \end{equation} 
    \end{minipage}
    \hfill\begin{minipage}{0.48\textwidth}
    \begin{equation}\label{mat:2} \left(\begin{array}{rrrr}
        \cellcolor{lightblue!120} 1 & \cellcolor{lightblue!120} 0 & \cellcolor{lightblue!120} -1 & \cellcolor{lightblue!120} 1 \\
        \cellcolor{lightblue!120} 0 & \cellcolor{lightblue!120} 1 & \cellcolor{lightblue!120} -1 & \cellcolor{lightblue!120} 1 \\
        -1 & 0 & 1 & -1
    \end{array}\right)\end{equation}
    \end{minipage}

    \noindent
    \begin{minipage}{0.48\textwidth}
    \begin{equation}\label{mat:3} \left(\begin{array}{rrrr}
        \cellcolor{yellow!75} -1 & \cellcolor{yellow!75} 1 & \cellcolor{lightpurple!60} 0 & \cellcolor{lightpurple!60} 0 \\
        \cellcolor{yellow!75} 1 & \cellcolor{yellow!75} -1 & \cellcolor{lightpurple!60} 0 & \cellcolor{lightpurple!60} 0 \\
        \cellcolor{yellow!75}0 & \cellcolor{yellow!75}0 & \cellcolor{lightpurple!60}-1 & \cellcolor{lightpurple!60}1
    \end{array}\right)\quad\end{equation} \end{minipage}
    \hfill
    \begin{minipage}{0.48\textwidth}
    \begin{equation}\label{mat:4} \left(\begin{array}{rrrr}
        \cellcolor{lightpurple!60} 0 & \cellcolor{lightpurple!60} 0 & \cellcolor{lightgreen!120} -1 & \cellcolor{lightgreen!120} 1 \\
        \cellcolor{lightpurple!60} 0 & \cellcolor{lightpurple!60} 0 & \cellcolor{lightgreen!120} 1 & \cellcolor{lightgreen!120} -1 \\
        \cellcolor{lightpurple!60}-1 & \cellcolor{lightpurple!60}1 & \cellcolor{lightgreen!120}1 & \cellcolor{lightgreen!120}-1
    \end{array}\right)\end{equation}
    \end{minipage}

    \noindent
    \begin{minipage}{\textwidth}
    \begin{equation}\label{mat:5} 
        \left(\begin{array}{rrrrrr}
        \cellcolor{yellow!75} -1 & \cellcolor{yellow!75} 1 & \cellcolor{coralorange!200} -1 & \cellcolor{coralorange!200} 1 & \cellcolor{lightgreen!120} -1 & \cellcolor{lightgreen!120} 1 \\
        \cellcolor{yellow!75} 1 & \cellcolor{yellow!75} -1 & \cellcolor{coralorange!200} 1 & \cellcolor{coralorange!200} -1 & \cellcolor{lightgreen!120} 1 & \cellcolor{lightgreen!120} -1 \\
        \cellcolor{yellow!75} 0 & \cellcolor{yellow!75} 0 & \cellcolor{coralorange!200} -1 & \cellcolor{coralorange!200} 1 & \cellcolor{lightgreen!120} 1 & \cellcolor{lightgreen!120} -1
    \end{array}\right)
    \end{equation}
    \end{minipage}
    \par\smallskip

\end{theorem}

Matrices \eqref{mat:1} and \eqref{mat:2}
exactly correspond to networks \hyperref[fig:1a]{(\subref{fig:1a})} and \hyperref[fig:1b]{(\subref{fig:1b})} in Fig.~\ref{fig:total_networks}. 
Consider the first two rows (blue entries) from \eqref{mat:1}–\eqref{mat:2}. Alongside the matrix resulting from removing the last column of the blue‑colored submatrix, they correspond to the degenerate $2$-species networks characterized in Lemma~\ref{lm:2d,2s,degnerate if and only if 1378=0}. Analogous to Michaelis-Menten kinetics \cite{MiMen1913}, the networks \hyperref[fig:1a]{(\subref{fig:1a})} and \hyperref[fig:1b]{(\subref{fig:1b})} model enzyme-catalyzed reactions exhibiting ACR in species $S_1$. In both topologies, an initial state $S_0$ is converted to the state $S_1$ (an enzyme-substrate complex) via catalyst $K$, followed by the spontaneous decay of $S_1$ back to $S_0$. Structurally, removing the ACR species $S_1$ from either network yields the same degenerate $2$-species network:
$0 \xrightarrow{} S_0,\;S_0+K \rightleftharpoons 0 \xrightarrow{} K$,
which constitutes the unique degenerate structure in $2$-species non-redundant networks. 
The fundamental distinction lies in their system invariants: network \hyperref[fig:1b]{(\subref{fig:1b})} exhibits a strict anti-coupling where $\dot{S}_1 = -\dot{S}_0$, effectively reducing the system dynamics; whereas network \hyperref[fig:1a]{(\subref{fig:1a})} involves a full-state conservation law $\dot{S}_0+\dot{S}_1=\dot{K}$, implying a tighter coupling among all $3$ species.
   It appears counterintuitive that  $S_1$ maintains ACR in network \hyperref[fig:1a]{(\subref{fig:1a})} despite its dynamics depending on other $2$ species, seemingly contradicting the notion that conservation laws hinder robustness.
 However, we demonstrate that this conservative topology is the sole exception in the $2$D non-redundant family ($s^*=3$). 
These claims are rigorously established later in Lemma \ref{lm:2d,3s,equival,has ACR in X1,X2,X3 i.f.f.} (i). 

    Notice that matrices \eqref{mat:3}--\eqref{mat:5} exactly correspond to Fig.~\hyperref[fig:1c]{\ref{fig:total_networks} (\subref{fig:1c})}, \hyperref[fig:1d]{(\subref{fig:1d})} and \hyperref[fig:1e]{(\subref{fig:1e})}.
     Networks \hyperref[fig:1c]{(\subref{fig:1c})}--\hyperref[fig:1e]{(\subref{fig:1e})} represent variations of the carbon nanotube rope model, characterized by an ACR species $X$. Structurally, removing $X$ from these networks reduces the dynamics of the remaining species ($C_0,\;C_e$) to a $1$D manifold. Physically, this corresponds to transitions between ground states of fractional or integral spin \cite{Cobden1998}, mediated by the absorption or release of a spin-$1/2$ electron ($X$). While the classic model \cite[Example 14]{Joshi2021} corresponds to the unique fully coupled subnetwork of \hyperref[fig:1e]{(\subref{fig:1e})} (obtained by removing blue-arrow reactions), we demonstrate that any consistent subnetwork of \hyperref[fig:1e]{(\subref{fig:1e})} inherently admits ACR. Furthermore, regarding structural robustness, introducing generic inflow-outflow reactions typically destroys this ACR property; persistence is guaranteed only if the system topology simplifies to that of \hyperref[fig:1c]{(\subref{fig:1c})} or \hyperref[fig:1e]{(\subref{fig:1e})}.
The formal mathematical treatment of these properties is detailed in Lemma
\ref{lm:2d,3s,equival,has ACR in X1,X2,X3 i.f.f.} (ii).

Theorem \ref{thm:2d,3s,3s^* with ACR table} establishes that any $s$-species ($s\geq 3$) network with non-vacuous ACR is a natural species refinement (Definition~\ref{def:species_refinement}) of a $3$-species network admitting non-vacuous ACR. In this construction, the $3$-species network serves as the lumping network; selected non-ACR species are replaced by sums of new species through a natural species refinement map $\phi$, while the ACR species is left unchanged. For instance, take the lumping network \hyperref[fig:1a]{(\subref{fig:1a})} and define $\phi$ by replacing $S_0$ and $K$ with the complexes $S_0+\sum_{i=1}^n S^{(i)}_0$ and $K+\sum_{i=1}^{\ell}K^{(i)}$, while fixing $S_1$:
\begin{equation*}S_0+\sum_{i=1}^n S^{(i)}_0+K+\sum_{i=1}^{\ell}K^{(i)} \rightleftharpoons 0 \xrightarrow{} K+\sum_{i=1}^{\ell}K^{(i)}+S_1,\;\; S_1 \xrightarrow{} S_0+\sum_{i=1}^n S^{(i)}_0.\end{equation*}
Crucially, each newly added species $X\in\{S^{(i)}_0,\;K^{(i)}\}$ inherits the stoichiometric row of its parent species in the lumping network, so that $\dot{X}\in\{\dot{S}_0,\;\dot{K}\}$. This is exactly the constraint imposed by a natural species refinement, and it ensures that the refined network preserves non-vacuous ACR in $S_1$. In contrast, introducing an intermediate species $SK$ with $\dot{SK}=-\dot{K}$ yields:
\begin{equation*}S_0+K \rightleftharpoons SK \xrightarrow{} K+S_1,\;\; S_1 \xrightarrow{} S_0.\end{equation*}
Here the stoichiometric matrix has four distinct rows, reflecting the four distinct types of concentration evolution $\dot{S}_0$, $\dot{K}$, $\dot{S}_1$, and $\dot{SK}$. The resulting network satisfies $s^*\geq 4$, so non-vacuous ACR in $S_1$ is precluded by Theorem~\ref{thm:2d,s*≥4 no ACR} in the next section.

\begin{remark}\label{rem:2d,3S*,ACR only in one speices}
  From Theorem \ref{thm:2d,3s,3s^* with ACR table} one can see that if $G$ has non-vacuous ACR, then $G$ has non-vacuous ACR only in one species, see the proof later in Corollary \ref{coro:2d,3s,N1N2-N1 has no ACR in X1 and X2}. 
   The same conclusion also holds for the $1$D case and the cases for $s^*=2$ in the previous section. 
\end{remark}

Below, we prepare lemmas for proving Theorem \ref{thm:2d,3s,3s^* with ACR table}.
\begin{lemma}\label{lm:2d,s,degenerate}
    Suppose $G$ is a $2$D non-redundant zero-one network with $s$ species. Assume  that the first two rows of the stoichiometric matrix $\mathcal{N}$ ($\mathcal{N}_1$, $\mathcal{N}_2$) are linearly independent. Then, $G$ is degenerate if and only if  the following two conditions hold simultaneously.   \begin{enumerate}
        \item [(i)] {$\begin{pmatrix}
            \mathcal{N}_1\\\mathcal{N}_2
        \end{pmatrix}$ is equivalent to $\left(\begin{array}{rrrr}
           \;\;1&\;\;0&-1&\;\;1\\\;\;0&\;\;1&-1&\;\;1
        \end{array}\right)$ or $\left(\begin{array}{rrr}
            \;\;1&\;\;0&-1\\\;\;0&\;\;1&-1
        \end{array}\right)$. }
        \item [(ii)] {$s^*=2$. }
    \end{enumerate}
\end{lemma}
\begin{proof}
    Let $G_{/3,\ldots,s}$ denote the network obtained by removing species $X_3$, $\cdots$, $X_s$ from $G$. Hence, the stoichiometric matrix of $G_{/3,\ldots,s}$ is composed of ${\mathcal N}_1$ and ${\mathcal N}_2$. Let $A(\lambda)$ and $A_{/3,\ldots,s}(\lambda)$ denote the partial transformed Jacobian matrix of $G$ and $G_{/3,\ldots,s}$, respectively.  
    
    ($\Leftarrow$) Since condition (i) holds, by Lemma \ref{lm:2d,2s,degnerate if and only if 1378=0}, $G_{/3,\ldots,s}$ is degenerate. Hence, by Lemma \ref{lm:degenerate equivals to A is zero poly}, we know that $det(A(\lambda)[\{1,2\},\;\{1,2\}])=det\left(A_{/3,\ldots,s}(\lambda)\right)$ is the zero polynomial. Since condition (ii) holds, for any $i,\;j\in\{1,\;\ldots,\;s\}$, either $\mathcal{N}_i=\mathcal{N}_j$ or $\{\mathcal{N}_i,\;\mathcal{N}_j\}=\{\mathcal{N}_1,\;\mathcal{N}_2\}$. If $\mathcal{N}_i=\mathcal{N}_j$, then $det(A(\lambda)[\{i,j\},\;\{i,j\}])$ is the zero polynomial. If $\{\mathcal{N}_i,\;\mathcal{N}_j\}=\{\mathcal{N}_1,\;\mathcal{N}_2\}$, then we have $det(A(\lambda)[\{i,j\},\;\{i,j\}])=\pm~det(A(\lambda)[\{1,2\},\;\{1,2\}])$. Therefore, by Lemma \ref{lm:degenerate equivals to A is zero poly}, $G$ is degenerate. 
    
    ($\Rightarrow$) Since $G$ is degenerate, by Lemma \ref{lm:degenerate equivals to A is zero poly},  $G_{/3,\ldots,s}$ is also degenerate. Hence, by Lemma \ref{lm:2d,2s,degnerate if and only if 1378=0},  condition (i) holds.
    Since $G$ is zero-one, for all $k\in\{3,\;\ldots,\;s\}$, \begin{equation*}
        \mathcal{N}_k\in\{\mathcal{N}_1,\;\mathcal{N}_2,\;-\mathcal{N}_1,\;-\mathcal{N}_2,\;\mathcal{N}_1-\mathcal{N}_2,\;-\mathcal{N}_1+\mathcal{N}_2\}. 
    \end{equation*}
    If condition (ii) does not hold, i.e., there exists $k\in\{3,\;\ldots,\;s\}$ such that $\mathcal{N}_k\notin\{\mathcal{N}_1,\;\mathcal{N}_2\}$, then it is easy to check that for any $i\in\{1,\;2\}$, $\left(\begin{smallmatrix}
        \mathcal{N}_i\\\mathcal{N}_k
    \end{smallmatrix}\right)$ is rank two and satisfies $\left(\begin{smallmatrix}
        \mathcal{N}_i\\\mathcal{N}_k
    \end{smallmatrix}\right)\notin\left\{\left(\begin{smallmatrix}
        \;\;1&\;\;0&-1&\;\;1\\\;\;0&\;\;1&-1&\;\;1
    \end{smallmatrix}\right),\;\left(\begin{smallmatrix}
        \;\;1&\;\;0&-1\\\;\;0&\;\;1&-1
    \end{smallmatrix}\right)\right\}$. Hence, by Lemma \ref{lm:2d,2s,degnerate if and only if 1378=0}, the $2$-species network  corresponding to the two rows $
        \mathcal{N}_i$ and $\mathcal{N}_k$ is nondegenerate. Hence, by Lemma \ref{lm:degenerate equivals to A is zero poly}, $G$ is nondegenerate, which leads to a contradiction.     
\end{proof}

\begin{lemma}\label{lm:if 2s degenerate, then 3-th species has ACR}
    Suppose $G$ is a $2$D non-redundant zero-one network with $s=s^*=3$. If the network obtained by removing species $X_3$ from $G$, say $G_{/3}$, is $2$D, then $G$ has non-vacuous ACR in species $X_3$ for any generic choice of rate constants if and only if $G_{/3}$ is degenerate. 
\end{lemma}
\begin{proof}
    ($\Rightarrow$) Assume that $G$ has non-vacuous ACR for any generic choice of rate constants in $X_3$. Notice that both $G_{/3}$ and $G$ are $2$D. Hence, by Theorem \ref{thm:main}, $G_{/3}$ is degenerate. 
    
    ($\Leftarrow$) Assume that $G_{/3}$ is degenerate. By Lemma \ref{lm:2d,s,degenerate}, if $G$ is degenerate, then $G$ satisfies $s^*=2$, which contradicts the assumption that $G$ satisfies $s^*=3$. Hence, $G$ is nondegenerate. By \cite[Theorem 3]{Jiao2025} and \cite[Theorem A]{Feliu2024}, for any generic choice of rate constants, $G$ has exactly one positive steady state. Then, by Remark \ref{rem:s-1 degenerate, x_i has ACR}, $G$ has non-vacuous ACR in species $X_3$ for any generic choice of rate constants. 
\end{proof}

In matrices \eqref{mat:3}--\eqref{mat:5}, the yellow, purple, orange, and green columns represent the pairs $\mathcal{B}_1$, $\mathcal{B}_2$, $\mathcal{B}_3$, and $\mathcal{B}_4$, respectively, where $\mathcal{B}_1=\left\{\begin{smallmatrix}-1\; &\;\;1\\
    \;\;1, & -1 \\\;\; 0\; &\;\;0\end{smallmatrix}\right\}$, $\mathcal{B}_2=\left\{\begin{smallmatrix}
        \;\;0\; & \;\;0\\
        \;\;0, & \;\;0\\
        -1\; & \;\;1        \end{smallmatrix}\right\}$, $\mathcal{B}_3=\left\{\begin{smallmatrix}
        -1\; & \;\;1\\
        \;\;1, & -1\\
        -1\; & \;\;1    \end{smallmatrix}\right\}$, and $\mathcal{B}_4=\left\{\begin{smallmatrix}
        -1\;&\;\;1\\\;\;1, & -1\\
        \;\;1\; & -1
\end{smallmatrix}\right\}$. 
    
\begin{lemma}\label{lm:if 2s 1d, then 3-th species has ACR}
    Suppose $G$ is a $2$D non-redundant zero-one network with $s=s^*=3$. Assume that  the network obtained by removing species $X_3$ from $G$, say  $G_{/3}$, is $1$D. Let $\mathcal{N}$ denote the stoichiometric matrix of $G$. 
      \begin{enumerate}
     \item [(i)] If $\mathcal{N}$ draws vectors from at least three of $\mathcal{B}_1$, $\mathcal{B}_2$, $\mathcal{B}_3$, and $\mathcal{B}_4$ with $\mathcal{B}_2$ included, then $G$ has no non-vacuous ACR in species $X_3$ for any generic choice of rate constants. 
     \item [(ii)] If $\mathcal{N}$ draws vectors from exactly two of $\mathcal{B}_1$, $\mathcal{B}_2$, $\mathcal{B}_3$, and $\mathcal{B}_4$ with $\mathcal{B}_2$ included, then either $G$ is not consistent, or $G$ has non-vacuous ACR in species $X_3$ for any generic choice of rate constants. 
     \item[(iii)] If $\mathcal{N}$ contains no vectors from $\mathcal{B}_2$, then  either $G$ is not consistent, or $G$ has non-vacuous ACR in species $X_3$ for any generic choice of rate constants. 
    \end{enumerate}
\end{lemma}
\begin{proof}
 Notice that $G$ is a $2$D non-redundant zero-one network with $3$ species. Without loss of generality, we assume that $\mathcal{N}_2=-\mathcal{N}_1$. Hence, $\mathcal{N}$ consists of at least $2$ linearly independent columns from the following matrix 
  \begin{equation}\label{vector pairs: B1234}
        \mathcal{N}_{max}\equals\left(\begin{array}{rrrrrrrr}
            -1&\;\;1&\;\;0&\;\;0&-1&\;\;1&-1&\;\;1\\
            \;\;1&-1&\;\;0&\;\;0&\;\;1&-1&\;\;1&-1\\
            \;\;0&\;\;0&-1&\;\;1&-1&\;\;1&\;\;1&-1
        \end{array}\right).
    \end{equation}
    Let $G_{max}$ be the non-redundant zero-one network corresponding to ${\mathcal N}_{max}$. Then,  the steady-state system augmented by the conservation law is \begin{equation}\label{sys:N1 N2 -N1}
        \begin{cases}
h_1 &= -\kappa_1x_1+\kappa_2x_2-\kappa_5x_1x_3+\kappa_6x_2-\kappa_7x_1+\kappa_8x_2x_3,\\
h_2 &=\;\;\;x_1+x_2-c,\\
h_3 &= -\kappa_3x_3+\kappa_4-\kappa_5x_1x_3+\kappa_6x_2+\kappa_7x_1-\kappa_8x_2x_3.
\end{cases}
    \end{equation}
    Notice that there exists no monomial  $x_1 x_2$ in \eqref{sys:N1 N2 -N1} (since ${\mathcal N}_1=-{\mathcal N}_2$). By collecting the terms containing $x_1$ and $x_2$, we can rewrite the steady-state equations $h_1=0$ and $h_3=0$ in the following form 
   \begin{equation}\label{sys:all with 34}
        \begin{cases}
        \frac{\partial h_1}{\partial x_2}~x_2 &= -\frac{\partial h_1}{\partial x_1}~x_1,\\
        \frac{\partial h_3}{\partial x_2}~x_2 &= -\frac{\partial h_3}{\partial x_1}~x_1+\kappa_3x_3-\kappa_4.  
        \end{cases}
    \end{equation} 
    We define \begin{equation}\label{al:defJx3}
    \mathcal{J}(\kappa,\;x_3) :=
    \begin{vmatrix}
        \dfrac{\partial h_1}{\partial x_1} & \dfrac{\partial h_1}{\partial x_2} \\[6pt]
        \dfrac{\partial h_3}{\partial x_1} & \dfrac{\partial h_3}{\partial x_2}
    \end{vmatrix}.
\end{equation}
    By \eqref{sys:N1 N2 -N1} and \eqref{al:defJx3}, it is easy to obtain 
        \begin{equation}\label{al:Jx3}
        \mathcal{J}=2\kappa_5\kappa_8~x_3^2+(\kappa_2\kappa_5+\kappa_1\kappa_8)~x_3-(\kappa_1\kappa_6+\kappa_2\kappa_7+2\kappa_6\kappa_7). 
    \end{equation}

    \textit{(i)} Suppose $\mathcal{N}$ is composed of vectors from at least three of $\mathcal{B}_1$, $\mathcal{B}_2$, $\mathcal{B}_3$, and $\mathcal{B}_4$. 
         Notice that if any vector in $\mathcal{B}_i$ ($i\in\{1,\;2,\;3,\;4\}$) is not contained in $\mathcal{N}$, then $\kappa_{2i-1}$ and $\kappa_{2i}$ vanish in $h$ and ${\mathcal J}$.  By \eqref{al:Jx3}, the equation $\mathcal{J}(\kappa,\;x_3)=0$ defines a function $x_3=x_3(\kappa)$, and 
         the function contains no variables in $\{\kappa_3,\;\kappa_4\}$.
         When $G$ is consistent, 
         we prove by contradiction that $G$ admits no non-vacuous ACR in species $X_3$ generically. Assume that there exists a generic choice of rate constants $\kappa^*$ such that $G$ has non-vacuous ACR in species $X_3$. Hence, we have $\mathcal{J}(\kappa^*,\;x_3(\kappa^*))=0$ for all positive steady states. By substituting $x_3=x_3(\kappa^*)$ into \eqref{sys:all with 34}, we have \begin{equation}\label{al:k3x3-x4}
            \kappa^*_3x_3(\kappa^*)-\kappa^*_4=0. 
        \end{equation}
        Notice that $\mathcal{N}$ contains vectors selected from $\mathcal{B}_2$. Hence, \eqref{al:k3x3-x4} contradicts the assumption that $\kappa^*$ is a generic choice of rate constants.

      \textit{(ii)}  Suppose $\mathcal{N}$ is composed of vectors from exactly two of $\mathcal{B}_1$, $\mathcal{B}_2$, $\mathcal{B}_3$, and $\mathcal{B}_4$ with the two selected sets containing $\mathcal{B}_2$. Then, the Jacobian matrix $\mathcal{J}(\kappa,\;x_3)$ defined in \eqref{al:Jx3} becomes rank-$1$, and by \eqref{sys:all with 34}, 
       we can obtain $x_3=\frac{\kappa_4}{\kappa_3}$. 
    Hence, $G$ has non-vacuous ACR in species $X_3$ for any generic choice of rate constants. 
    
      \it{(iii)} 
       Suppose $\mathcal{N}$ contains no vectors from $\mathcal{B}_2$. Hence, the reactions assigned to $\kappa_3$ and $\kappa_4$ do not appear in the network $G$. So, \eqref{sys:all with 34} becomes \begin{equation}\label{sys:all without 34}
        \begin{cases}
        \frac{\partial h_1}{\partial x_2}~x_2 &= -\frac{\partial h_1}{\partial x_1}~x_1,\\
        \frac{\partial h_3}{\partial x_2}~x_2 &= -\frac{\partial h_3}{\partial x_1}~x_1.
        \end{cases}
    \end{equation}
    Hence, by \eqref{al:defJx3} and \eqref{sys:all without 34}, we have that $\mathcal{J}(\kappa,\;x_3)$ vanishes at all positive steady states. From \eqref{al:Jx3}, one can see that $\mathcal{J}(\kappa,\;x_3)=0$ has at most one positive solution in $x_3$ ($x_3=x_3(\kappa)$) for any $(\kappa_1,\;\kappa_2,\;\kappa_5,\;\kappa_6,\;\kappa_7,\;\kappa_8)$ $\in {\mathbb R}^6_{\geq 0}$. Hence, if $G$ is consistent, then $G$ has non-vacuous ACR in species $X_3$ for any generic choice of rate constants.  
\end{proof}

\begin{lemma}\label{lm:2d,3s,equival,has ACR in X1,X2,X3 i.f.f.}
    Suppose $G$ is a $2$D non-redundant zero-one network with $s=s^*=3$. Let $G_{/i}$ be the network obtained by removing a single species $X_i$ from $G$. 
    Then $G$ has non-vacuous ACR in species $X_i$ for any generic choice of rate constants if and only if
    there exist a $2$D non-redundant zero-one network $\widetilde{G}$ with $3$ species $\widetilde{X}_1$, $\widetilde{X}_2$, $\widetilde{X}_3$ and a consistent subnetwork $\widetilde{G}_{\mathrm{sub}}$ of $\widetilde{G}$ such that $G$ is equivalent to $\widetilde{G}_{\mathrm{sub}}$ and $X_i$ in  $G$ corresponds to $\widetilde{X}_3$ in $\widetilde{G}_{\mathrm{sub}}$, 
    where the stoichiometric matrix $\widetilde{\mathcal{N}}$ of $\widetilde{G}$ is one of \eqref{mat:1}--\eqref{mat:5}.

     Furthermore,
 \begin{enumerate}
     \item[(i)]  if $G_{/i}$ is $2$D, then the above $\widetilde{G}$ corresponds to one of \eqref{mat:1}–\eqref{mat:2}, and 
     \item[(ii)] if $G_{/i}$ is $1$D, then the above $\widetilde{G}$ corresponds to one of \eqref{mat:3}--\eqref{mat:5}. 
 \end{enumerate}
    
\end{lemma}
\begin{proof}  
 ($\Rightarrow$) Without loss of generality,  assume that $i=3$. First, if $G_{/3}$ is $2$D,  then by Lemma \ref{lm:if 2s degenerate, then 3-th species has ACR}, $G_{/3}$ is degenerate. Let ${\mathcal N}_{/3}$ denote the stoichiometric matrix of $G_{/3}$. By Lemma \ref{lm:2d,2s,degnerate if and only if 1378=0}, we know that ${\mathcal N}_{/3}$ is equivalent to $\left(\begin{smallmatrix}
       \;\;1&\;\;0&-1&\;\;1\\\;\;0&\;\;1&-1&\;\;1
   \end{smallmatrix}\right)$ or $\left(\begin{smallmatrix}
            \;\;1&\;\;0&-1\\\;\;0&\;\;1&-1
        \end{smallmatrix}\right)$.
    Notice that $s^*=3$. Assume that $\mathcal{N}_3=a\mathcal{N}_1+b\mathcal{N}_2$. Hence, either $a$ and $b$ satisfy $ab\neq0$, or $(a,\;b)$ equals $(-1,\;0)$ or $(0,\;-1)$. 
    If $ab\neq0$, then $\mathcal{N}$ must be equivalent to a submatrix of \eqref{mat:1}, say $\widetilde{\mathcal{N}}_{\mathrm{sub}}$. Notice that the network  corresponding to $\widetilde{\mathcal{N}}_{\mathrm{sub}}$, say $\widetilde{G}_{\mathrm{sub}}$, admits non-vacuous ACR in its $3$rd species (say $\widetilde{X}_3$) generically. Hence, $X_3$ in $G$ corresponds to $\widetilde{X}_3$ in $\widetilde{G}_{\mathrm{sub}}$.  
    If $(a,\;b)=(-1,\;0)$ or $(a,\;b)=(0,\;-1)$, then $\mathcal{N}$ must be equivalent to a submatrix of \eqref{mat:2}, say $\widetilde{\mathcal{N}}_{\mathrm{sub}}$.  Hence, $G$ is equivalent to  the network  corresponding to $\widetilde{\mathcal{N}}_{\mathrm{sub}}$, say $\widetilde{G}_{\mathrm{sub}}$, and $X_3$ in  $G$ corresponds to $\widetilde{X}_3$ in $\widetilde{G}_{\mathrm{sub}}$. 

    Second, if $G_{/3}$ is $1$D, then by  Lemma \ref{lm:if 2s 1d, then 3-th species has ACR}, 
    $\mathcal{N}$ must be equivalent to a submatrix of \eqref{mat:3}--\eqref{mat:5}, say $\widetilde{\mathcal{N}}_{\mathrm{sub}}$. So, $G$ is equivalent to  the network $\widetilde{G}_{\mathrm{sub}}$ corresponding to $\widetilde{\mathcal{N}}_{\mathrm{sub}}$, and $X_3$ in  $G$ corresponds to $\widetilde{X}_3$ in $\widetilde{G}_{\mathrm{sub}}$.  
    
    ($\Leftarrow$) If $\widetilde{G}$ corresponds to one of \eqref{mat:1}--\eqref{mat:2}, then by Lemma \ref{lm:2d,2s,degnerate if and only if 1378=0} and Lemma \ref{lm:if 2s degenerate, then 3-th species has ACR}, any consistent subnetwork $\widetilde{G}_{\mathrm{sub}}$   admits non-vacuous ACR in species $\widetilde{X_3}$ generically. Since $G$ is equivalent to $\widetilde{G}_{\mathrm{sub}}$ and $X_i$ in $G$ corresponds to $\widetilde{X}_3$ in $\widetilde{G}_{\mathrm{sub}}$, $G$ has non-vacuous ACR in species $X_i$ for any generic choice of rate constants. 
    If $\widetilde{G}$ corresponds to one of  \eqref{mat:3}--\eqref{mat:5}, then the conclusion similarly holds by Lemma \ref{lm:if 2s 1d, then 3-th species has ACR}. 
\end{proof}

\begin{remark}\label{rem:consistent subnetworks mat1-4}
All stoichiometric matrices for consistent subnetworks derived from \eqref{mat:1}--\eqref{mat:5} are compiled in Table~\ref{tab:stoich_3s_3s*_ACR}. 
\begin{table}[!tbh]
\centering
\caption{Stoichiometric matrices characterizing non‑vacuous ACR networks}
\label{tab:stoich_3s_3s*_ACR}
\small
\begin{tabular}{ll}
\toprule
\refstepcounter{matnum}\hypertarget{rommat:\arabic{matnum}}{}\textbf{(\thematnum)} $\;\;\left(\begin{array}{rrrr}
     1 &\;\; 0 & -1 &\;\; 1 \\ 0 & 1 & -1 &  1 \\
    -1 & 1 & 0 & 0
\end{array}\right)$
&
\refstepcounter{matnum}\hypertarget{rommat:\arabic{matnum}}{}\textbf{(\thematnum)} $\;\;\left(\begin{array}{rrr}
      1 &  \;\;0 &  -1 \\ 0 & 1 & -1 \\
    -1 & 1 & 0
\end{array}\right)$ \\[1.2em]

\refstepcounter{matnum}\hypertarget{rommat:\arabic{matnum}}{}\textbf{(\thematnum)} $\;\;\left(\begin{array}{rrrr}
     1 &\;\; 0 & -1 & 1 \\ 0 &  1 &  -1 &  1 \\
    -1 & 0 & 1 & -1
\end{array}\right)$
&
\refstepcounter{matnum}\hypertarget{rommat:\arabic{matnum}}{}\textbf{(\thematnum)} $\;\;\left(\begin{array}{rrr}
     1 &\;\; 0 &   -1 \\ 0 & 1 & -1 \\
    -1 & 0 & 1
\end{array}\right)$ \\

\refstepcounter{matnum}\hypertarget{rommat:\arabic{matnum}}{}\textbf{(\thematnum)} $\;\;\left(\begin{array}{rrrr}
    -1 & 1 & 0 &\;\; 0 \\ 1 &  -1 & 0 &  0 \\0 & 0 & -1 & 1
\end{array}\right)$
&
\refstepcounter{matnum}\hypertarget{rommat:\arabic{matnum}}{}\textbf{(\thematnum)} $\;\;\left(\begin{array}{rrrr}
     0 &\;\;  0 &  -1 &  1 \\ 0 &  0 & 1 & -1 \\-1 & 1 & 1 & -1
\end{array}\right)$ \\[1.2em]

\refstepcounter{matnum}\hypertarget{rommat:\arabic{matnum}}{}\textbf{(\thematnum)} $\;\;\left(\begin{array}{rrr} -1 & 1 & 1 \\ 1 & -1 & -1 \\ 0 & 1 & -1 \end{array}\right)$
&
\refstepcounter{matnum}\hypertarget{rommat:\arabic{matnum}}{}\textbf{(\thematnum)} $\;\;\left(\begin{array}{rrrr} -1 & 1 & -1 & 1 \\ 1 & -1 & 1 & -1 \\ 0 & 0 & -1 & 1 \end{array}\right)$ \\[1.2em]

\refstepcounter{matnum}\hypertarget{rommat:\arabic{matnum}}{}\textbf{(\thematnum)} $\;\;\left(\begin{array}{rrrr} -1 & 1 & -1 & 1 \\ 1 & -1 & 1 & -1 \\ -1 & 1 & 1 & -1 \end{array}\right)$
&
\refstepcounter{matnum}\hypertarget{rommat:\arabic{matnum}}{}\textbf{(\thematnum)} $\left(\begin{array}{rrrr} -1 & 1 & -1 & -1 \\ 1 & -1 & 1 & 1 \\ 0 & 0 & -1 & 1 \end{array}\right)$ \\[1.2em]

\refstepcounter{matnum}\hypertarget{rommat:\arabic{matnum}}{}\textbf{(\thematnum)} $\left(\begin{array}{rrrr} -1 & -1 & 1 & 1 \\ 1 & 1 & -1 & -1 \\ 0 & -1 & 1 & -1 \end{array}\right)$
&
\refstepcounter{matnum}\hypertarget{rommat:\arabic{matnum}}{}\textbf{(\thematnum)} $\left(\begin{array}{rrrr} -1 & 1 & -1 & 1 \\ 1 & -1 & 1 & -1 \\ 0 & 1 & 1 & -1 \end{array}\right)$ \\[1.2em]

\refstepcounter{matnum}\hypertarget{rommat:\arabic{matnum}}{}\textbf{(\thematnum)} $\left(\begin{array}{rrrrr} -1 & 1 & -1 & -1 & 1 \\ 1 & -1 & 1 & 1 & -1 \\ 0 & 0 & -1 & 1 & -1 \end{array}\right)$
&
\refstepcounter{matnum}\hypertarget{rommat:\arabic{matnum}}{}\textbf{(\thematnum)} $\left(\begin{array}{rrrrr} -1 & 1 & 1 & -1 & 1 \\ 1 & -1 & -1 & 1 & -1 \\ 0 & 0 & 1 & 1 & -1 \end{array}\right)$ \\[1.2em]

\refstepcounter{matnum}\hypertarget{rommat:\arabic{matnum}}{}\textbf{(\thematnum)} $\left(\begin{array}{rrrrr} -1 & -1 & 1 & -1 & 1 \\ 1 & 1 & -1 & 1 & -1 \\ 0 & -1 & 1 & 1 & -1 \end{array}\right)$
&
\refstepcounter{matnum}\hypertarget{rommat:\arabic{matnum}}{}\textbf{(\thematnum)} $\left(\begin{array}{rrrrrr}
     -1 &  1 &  -1 &  1 &  -1 &  1 \\
     1 & -1 & 1 & -1 &  1 &  -1 \\
     0 & 0 & -1 & 1 & 1 & -1
\end{array}\right)$ \\
\bottomrule
\end{tabular}
\end{table}
Specifically, \matref{1}--\matref{2} correspond to the subnetworks of \eqref{mat:1}, \matref{3}--\matref{4} to those of \eqref{mat:2}, \matref{5} and \matref{6} to \eqref{mat:3} and \eqref{mat:4}, respectively, and \matref{7}--\matref{16} to the subnetworks originating from \eqref{mat:5}. 
\end{remark}

\begin{corollary}\label{coro:2d,3s,N1N2-N1 has no ACR in X1 and X2}
     If $\mathcal{N}$, the stoichiometric matrix of a non‑redundant zero‑one network $G$, is one of \matref{1}--\matref{16}, then $G$ has non‑vacuous ACR only in species $X_3$ for any generic choice of rate constants. 
\end{corollary}
\begin{proof}
    Lemma \ref{lm:2d,3s,equival,has ACR in X1,X2,X3 i.f.f.} indicates that $G$ has non-vacuous ACR in species $X_3$ for any generic choice of rate constants if and only if by permuting the columns, $\mathcal{N}$ is the stoichiometric matrix of one of the consistent subnetworks derived from \eqref{mat:1}--\eqref{mat:5}. If $\mathcal{N}$ is one of \matref{1}--\matref{16}, then by Remark \ref{rem:consistent subnetworks mat1-4}, $\mathcal{N}$ is a submatrix of \eqref{mat:1}--\eqref{mat:5} that yields a consistent network. 
    Notice that no two matrices among \matref{1}--\matref{16} can be equivalent by permuting their rows. So,  $G$ has non-vacuous ACR in
    neither $X_1$ nor $X_2$ for any generic choice of rate constants. 
\end{proof}
For each network corresponding to the stoichiometric matrices \matref{1}--\matref{16}, one can  solve the steady‑state system over ${\mathbb Q}(\kappa)[x]$ and verify that the corresponding network has non-vacuous ACR  in species $X_3$ for any generic choice of rate constants. The  supporting codes are available online\footnotemark \footnotetext{\quad\url{https://github.com/Cynthia-091225/Detect-for-ACR/blob/main/verification.mw}.}.

\begin{lemma}\label{lm:2d,s,3s^*}
    Suppose $G$ is a $2$D non-redundant zero-one network with $s$ species ($s\geq4$) and $s^*=3$. Let $G_{/4,\ldots,s}$ denote the network obtained by removing species $X_4$, $\cdots$, $X_s$ from $G$. Assume that $G_{/4,\ldots,s}$ is still $2$D and satisfies $s^*=3$.  If for an index $i\in\{1,\;2,\;3\}$, we have
$\mathcal{N}_k\neq\mathcal{N}_i$ in the stoichiometric matrix $\mathcal{N}$ of $G$ for all $k\in\{1,\;\ldots,\;s\}\setminus \{i\}$, then for any generic choice of rate constants, 
$G$ has non-vacuous ACR  in species $X_i$ if and only if 
    $G_{/4,\ldots,s}$ has non-vacuous ACR  in species $X_i$. 
   
\end{lemma}
\begin{proof}
   Since $G_{/4,\ldots,s}$ still satisfies $s^*=3$, 
   the first $3$ rows of ${\mathcal N}$ are pairwise distinct. 
   Without loss of generality,  assume that $i=3$.
   Define $\delta_j\;:=\;\{k\;|\;\mathcal{N}_k=\mathcal{N}_j,\;k\in\{1,\;\ldots,\;s\}\}$, where $j\in\{1,\;2\}$. Let $Y_1=\prod_{k\in\delta_1}x_k$ and $Y_2=\prod_{k\in\delta_2}x_k$. Let $f$ and $f_{/4,\ldots,s}$ denote the steady-state systems of $G$ and $G_{/4,\ldots,s}$, respectively. Then the conclusion follows from the fact that $(f_1,\;f_2,\;f_3)=f_{/4,\ldots,s}|_{x_1=Y_1,\;x_2=Y_2}$. 
\end{proof}

\noindent\textbf{\textit{Proof of Theorem \ref{thm:2d,3s,3s^* with ACR table}}}
Without loss of generality, assume that $i=3$ (i.e., the ACR species we are investigating is $X_3$), and assume that 
the first $3$ rows of $\mathcal{N}$ are pairwise distinct since $s^*=3$.

 ($\Leftarrow$) Suppose there exist a $2$D non-redundant zero-one network $\widetilde{G}$ with $3$ species $\widetilde{X}_1$, $\widetilde{X}_2$, $\widetilde{X}_3$, a consistent subnetwork $\widetilde{G}_{\mathrm{sub}}$ of $\widetilde{G}$, and a natural species refinement map $\phi$ such that $G=\phi(\widetilde{G}_{\mathrm{sub}})$ and $\phi(\widetilde{X}_3)=X_3$, where the stoichiometric matrix $\widetilde{\mathcal{N}}$ of $\widetilde{G}$ is one of \eqref{mat:1}--\eqref{mat:5}. By Lemma \ref{lm:2d,3s,equival,has ACR in X1,X2,X3 i.f.f.}, the network $\widetilde{G}_{\mathrm{sub}}$ has non-vacuous ACR in species $\widetilde{X}_3$ for any generic choice of rate constants. Since $G=\phi(\widetilde{G}_{\mathrm{sub}})$ and $\phi(\widetilde{X}_3)=X_3$,  the network $G_{/4,\ldots,s}$ is equivalent to  $\widetilde{G}_{\mathrm{sub}}$ and $X_3$ in $G_{/4,\ldots,s}$ corresponds to $\widetilde{X}_3$ in $\widetilde{G}_{\mathrm{sub}}$. So, $G_{/4,\ldots,s}$  has non-vacuous ACR in species $X_3$ generically. By Lemma \ref{lm:2d,s,3s^*}, the same conclusion holds for $G$, which completes the proof.
    
    ($\Rightarrow$) Suppose that $G$ admits non-vacuous ACR in species $X_3$ for any generic choice of rate constants. By Lemma \ref{lm: symmetry for no ACR}, for all $k\in\{1,\;\ldots,\;s\}\setminus\{3\}$, $\mathcal{N}_k\in\{\mathcal{N}_1,\;\mathcal{N}_2\}$. By Lemma \ref{lm:2d,s,3s^*}, the $3$-species network $G_{/4,\ldots,s}$ also admits non-vacuous ACR in species $X_3$ generically. 
    By Lemma \ref{lm:2d,3s,equival,has ACR in X1,X2,X3 i.f.f.}, 
    there exist a $2$D non-redundant zero-one network $\widetilde{G}$ with $3$ species $\widetilde{X}_1$, $\widetilde{X}_2$, $\widetilde{X}_3$ and a consistent subnetwork $\widetilde{G}_{\mathrm{sub}}$ of $\widetilde{G}$ such that $G_{/4,\ldots,s}$ is equivalent to $\widetilde{G}_{\mathrm{sub}}$ and $X_3$ in  $G_{/4,\ldots,s}$ corresponds to $\widetilde{X}_3$ in $\widetilde{G}_{\mathrm{sub}}$, 
    where the stoichiometric matrix $\widetilde{\mathcal{N}}$ of $\widetilde{G}$ is one of \eqref{mat:1}--\eqref{mat:5}.  Define a natural species refinement map $\phi$ on the species set $\{\widetilde{X}_1,\;\widetilde{X}_2,\;\widetilde{X}_3\}$ of  $\widetilde{G}_{\mathrm{sub}}$, by \begin{equation*}
        \phi(\widetilde{X}_p)=\sum_{k\in I_p}X_k,\quad p=1,\;2,\qquad \phi(\widetilde{X}_3)=X_3, 
    \end{equation*}
    where $I_p\equals\{k|\mathcal{N}_k=\mathcal{N}_p,\;k\in\{1,\;\ldots,\;s\}\setminus\{3\}\}$ and $\mathcal{N}_{\ell}$ denotes the $\ell$-th row of the stoichiometric matrix of $G$. Then,  $G=\phi(\widetilde{G}_{\mathrm{sub}})$ is a natural species refinement of $\widetilde{G}_{\mathrm{sub}}$ with $\phi(\widetilde{X}_3)=X_3$, which completes the proof. 
    \hfill$\square$

\subsubsection{Networks with $\mathbf{s^*\geq4}$}\label{sssec:2d,s^*>=4}

\begin{theorem}\label{thm:2d,s*≥4 no ACR}
    If $G$ is a $2$D non-redundant zero-one network with  $s^*\geq4$, then $G$ has no non-vacuous ACR for any generic choice of rate constants. 
\end{theorem}

\noindent\textbf{\textit{Proof of Theorem \ref{thm:2d,s*≥4 no ACR}}}\quad
For any $i\in\{1,\;\ldots,\;s\}$, let $G_{/i}$ denote the network obtained by removing species $X_i$ from $G$. Since $G$ satisfies $s^*\geq4$, $G_{/i}$ satisfies $s^*\geq3$. Notice that $G_{/i}$ must be $2$D, and by Lemma \ref{lm:2d,s,degenerate}, $G_{/i}$ is nondegenerate. Hence, by Theorem \ref{thm:main}, $G$ has no non-vacuous ACR in species $X_i$ for  generic  rate constants. Therefore, $G$ admits no non-vacuous ACR generically. \hfill$\square$
    
\begin{remark}
    Suppose $G$ is a $2$D network with $s$ species and $s^*\geq4$. Notice that $s^*\geq4$ indicates that there are no less than four distinct rows in the stoichiometric matrix $\mathcal{N}$. In other words, among the $s-2$ conservation laws, we have at least two conservation laws that are not in the form of $x_i=x_j+c$ ($i,\;j\in\{1,\;\ldots,\;s\}$). 
\end{remark}

\section{Discussion}\label{sec:discussion}

For non-redundant zero-one networks of dimensions higher than two, we conjecture that if $s^*$ exceeds a certain value $s^*_{max}(s, r)$, such networks have no non-vacuous ACR for any generic choice of rate constants.  One only needs to show  that a non-redundant zero-one network is always nondegenerate once $s^*$ reaches $s^*_{max}-1$.






\section*{Acknowledgments}

The authors were supported by the National Natural Science Foundation of China (NSFC No. 12671550). 

\bigskip
\bibliography{sn-bibliography}

\begin{appendices}

\section{A self-contained proof of Theorem \ref{thm:main}}\label{app:proof}
\textbf{Proof. } Without loss of generality, we assume that the removed species is $X_s$. Let $\mathcal{N}$ and $\mathcal{N}_{/s}$ denote the stoichiometric matrices of $G$ and $G_{/s}$, respectively. Notice that the flux cones of $\mathcal{N}$ and $\mathcal{N}_{/s}$ are the same, since ${\rm rank}(\mathcal{N})={\rm rank}(\mathcal{N}_{/s})=r$ and $\mathcal{N}_{/s}$ is obtained by deleting the $s$-th row of $\mathcal{N}$. Let $A(\lambda)$ and $A_{/s}(\lambda)$ denote the partial transformed Jacobian matrices of $G$ and $G_{/s}$, respectively. Notice that \begin{equation}\label{al:tildeA=As-1}
       A_{/s}(\lambda)=A(\lambda)[\{1,\ldots,s-1\},\;\{1,\ldots,s-1\}].\;
   \end{equation}
   Since $G_{/s}$ is nondegenerate, by Lemma \ref{lm:degenerate equivals to A is zero poly}, there exists $I^*\subseteq\{1,\ldots,s-1\}$ satisfying $|I^*|=r$ such that $det(A_{/s}[I^*,\;I^*])$ is not the zero polynomial. Hence, by \eqref{al:tildeA=As-1}, $det(A[I^*,\;I^*])$ is not the zero polynomial. Let $J$ denote the transformed Jacobian matrix of $G$. Define \begin{equation}\label{eq:tildeB}
       B_{/s}(p,\;\lambda)\equals\sum_{\substack{I\subseteq\{1,\;\ldots,\;s-1\}\\|I|=r}}det(J[I,\;I]),\;\;
B(p,\;\lambda)\equals\sum_{\substack{I\subseteq\{1,\;\ldots,\;s\}\\|I|=r}}det(J[I,\;I]).
   \end{equation}
   Recall that $J(p,\;\lambda)=A(\lambda)~diag(p)$.  So, neither $B_{/s}(p,\;\lambda)$ nor $B(p,\;\lambda)$ is the zero polynomial. 
   Let \begin{equation}\label{al:thm m}
       \Gamma_j(\kappa,\;p,\;\lambda)\equals\kappa_j-\prod_{i=1}^s p_i^{\alpha_{ji}}\sum_{k=1}^n\lambda_k\ell_j^{(k)},\quad j=1,\;\ldots,\;m.\;
   \end{equation}
   Consider the following $m+2$ equations, \begin{equation*}\begin{cases}
           &B_{/s}(p,\;\lambda)=0,\;\\
           &B(p,\;\lambda)=0,\;\\
           &\Gamma_j(\kappa,\;p,\;\lambda)=0,\;j=1,\;\ldots,\;m.
       \end{cases} \end{equation*} 
   Let $\mathcal{I}_1\equals\langle B,\;\Gamma_1,\;\ldots,\;\Gamma_m\rangle$ and $\mathcal{I}_2\equals\langle B_{/s},\;\Gamma_1,\;\ldots,\;\Gamma_m\rangle$ denote the ideals generated by the corresponding polynomials in $\mathbb{Q}[\kappa]$. Let $V(\mathcal{I}_1)$ and $V(\mathcal{I}_2)$ denote the algebraic varieties corresponding to $\mathcal{I}_1$ and $\mathcal{I}_2$, respectively. Let $\pi:\mathbb{C}^m\times\mathbb{C}^s\times\mathbb{C}^n\rightarrow\mathbb{C}^m$ be the regular projection mapping. Hence, we have $\pi(V(\mathcal{I}_i))\subseteq\mathbb{C}^m$ ($i\in\{1,\;2\}$). 
   Let $\Theta_i\equals\{g\in\mathbb{Q}[\kappa]\;\vert\;\forall b\in\pi(V(\mathcal{I}_i)),\;g(b)=0\}$. 
   Define 
\begin{equation}\label{eq:Sinproof}
S_i\equals\overline{\pi\left(V(\mathcal{I}_i)\right)}=\left\{a\in\mathbb{C}^m\;\vert\;\forall g\in \Theta_i,\;g(a)=0\right\}\subseteq\mathbb{C}^m.  
  \end{equation}
   Notice that we have $S_1\subsetneq\mathbb{C}^m$ and $S_2\subsetneq\mathbb{C}^m$ since $B_{/s}(p,\;\lambda)$ and $B(p,\;\lambda)$  are both non-zero polynomials. Therefore, $\mathbb{R}_{>0}^m\setminus(S_1\cup S_2)\neq\emptyset$. For any $\kappa^*\in\mathbb{R}_{>0}^m\setminus(S_1\cup S_2)$, by \eqref{al:thm m}, we can find a corresponding $(p^*,\;\lambda^*)$ such that $\Gamma_j(\kappa^*,\;p^*,\;\lambda^*)=0$ for all $j\in\{1,\;\ldots,\;m\}$. By the definition of $S_1$ and $S_2$ in \eqref{eq:Sinproof}, we have $B_{/s}(p^*,\;\lambda^*)\neq0$ and $B(p^*,\;\lambda^*)\neq0$ (otherwise, $(p^*,\;\lambda^*)\in S_1\cup S_2$). For the network $G$, let $h$ denote the steady-state system augmented by the conservation laws, and let $f$ denote the steady-state system. By \cite[Proposition 5.3]{Wiuf2013}, we have \begin{equation}\label{al:thm detJach=sumdetJacf}
       det(Jac_h)=\sum_{\substack{I\subseteq\{1,\;\ldots,\;s\}\\|I|=r}}det(Jac_f[I,\;I]).\;
   \end{equation}
   For the pair $(p^*,\;\lambda^*)$, let $x^*=\left(\frac{1}{p_1^*},\ldots,\frac{1}{p_s^*}\right)$. Notice that \begin{equation}\label{al:thm Jacfkx=Jplambda}
       Jac_f(\kappa^*,\;x^*)=J(p^*,\;\lambda^*).\;
   \end{equation}
   Hence, by \eqref{eq:tildeB}, \eqref{al:thm detJach=sumdetJacf} and \eqref{al:thm Jacfkx=Jplambda}, we have \begin{equation}\label{al:h and B}
       det(Jac_h)(\kappa^*,\;x^*)=B(p^*,\;\lambda^*)\neq0.\;
   \end{equation}
   Assume that the first $r$ rows of $\mathcal{N}$ are linearly independent. By the definition of $h$ in \eqref{eq:h_eta}, $h_{r+1}$, $\cdots$, $h_s$ represent the conservation laws. By substituting $x^*$ into conservation laws given by $h_{r+1},\;\ldots,\;h_s$, we can obtain a corresponding total-constant vector, denoted by $c^*=(c_{r+1}^*,\ldots,c_s^*)$. By \eqref{al:h and B}, notice that $x^*$ is a nondegenerate positive steady state. In some neighborhoods of $(\kappa^*,\;c^*)$, by the implicit function theorem, $x=(x_1,\ldots,x_s)$ can be regarded as a function of $\kappa$ and $c$. By Cramer's rule, we have \begin{equation}\label{al:xs partial cs neq 0}
        \frac{\partial x_s}{\partial c_s}(\kappa^*,\;c^*)     =\frac{(-1)^{2s}~\frac{\partial h_s}{\partial x_s}~ det(Jac_h[\{1,\ldots,s-1\},\;\{1,\ldots,s-1\}])}{det(Jac_h)}(\kappa^*,\;x^*).\;
    \end{equation}
    Again, by \cite[Proposition 5.3]{Wiuf2013}, we have \begin{equation}\label{al:detJachs-1=sumdetJacfs-1}
        det(Jac_h[\{1,\ldots,s-1\},\;\{1,\ldots,s-1\}]) =\sum_{\substack{I\subseteq\{1,\;\ldots,\;s-1\}\\|I|=r}}det(Jac_f[I,\;I]).\;
    \end{equation}
    Hence, by \eqref{eq:tildeB}, \eqref{al:thm Jacfkx=Jplambda} and \eqref{al:detJachs-1=sumdetJacfs-1}, we have \begin{equation}\label{al:thm detJachs-1neq0}
       det(Jac_h[\{1,\ldots,s-1\},\;\{1,\ldots,s-1\}])(\kappa^*,\;x^*)=B_{/s}(p^*,\;\lambda^*)\neq0.\;
   \end{equation}
    Notice that $h_s$ can be written as $h_s=x_s-\sum_{j=1}^r\gamma_jx_j-c_s$. Hence, $\frac{\partial h_s}{\partial x_s}\neq0$. So, by \eqref{al:xs partial cs neq 0} and \eqref{al:thm detJachs-1neq0},  we have $\frac{\partial x_s}{\partial c_s}(\kappa^*,\;c^*)\neq0$. Therefore, the network has no non-vacuous ACR in species $X_s$  for any $\kappa^*\in\mathbb{R}_{>0}^m\setminus(S_1\cup S_2)$. Hence, we conclude that $G$ has no non-vacuous ACR in species $X_s$ for any generic choice of rate constants. \hfill$\square$
 
\section{Algorithm and Experiments}\label{subsec:alg}

In this section, we record an algorithm based on the criterion in Section~\ref{sec:criterion} to detect non‑vacuous ACR, see Algorithm~\ref{alg:main}. The supporting code, are available at \\
\url{https://github.com/Cynthia‑091225/Detect‑for‑ACR}. The ACR‑related results and runtime summary for applying Algorithm~\ref{alg:main} to the biochemical networks shown in Fig.~\ref{fig:networks_background} are provided in Table~\ref{tab:algorithm summary}, which demonstrate the efficiency of Algorithm~\ref{alg:main}.

\begin{algorithm}
    \caption{Detecting non-vacuous ACR}\label{alg:main}
    \begin{description}
        \item[Input.] A network $G$ 
        \item[Output.] “Yes, $X_i$” or “No”
        \item[Step 1.] {Compute  
        $A(\lambda)\equals \mathcal{N}\;diag\left(\sum_{j=1}^n \lambda_j\ell^{(j)}\right)\; \mathcal{Y}^{\top}$. }
        \item[Step 2.] {For each $i\in \{1,\;\cdots,\;s\}$, do the following steps. }
        \begin{description}
            \item[{Step 2.1}] {Let $\mathcal{N}_{/i}$ denote the matrix obtained by deleting the $i$-th row of $\mathcal{N}$. }
            \item[{Step 2.2}] If ${\rm rank}\left(\mathcal{N}_{/i}\right)=r-1$, then compute \[\mathcal{L}_{/i}\equals\left\{\det\left(A(\lambda)[\{1,\ldots,r\},\;I]\right)\;|\;I\subseteq \{1,\;\cdots,\;s\}\setminus\{i\}\;\text{and}\;|I|=r\right\}. \]
            \hspace*{\dimexpr\labelwidth+4\labelsep\relax} If ${\rm rank}\left(\mathcal{N}_{/i}\right)=r$, then compute 
            \[\mathcal{L}_{/i}\equals\left\{\det\left(A(\lambda)[I,\;I]\right)\;|\;I\subseteq \{1,\;\cdots,\;s\}\setminus\{i\}\;\text{and}\;|I|=r\right\}. \]
        \end{description}
    \item[Step 3.] If $\mathcal{L}_{/i}=\{0\}$ for some $i$, then return “Yes, $X_i$”. \\\hspace*{\dimexpr\labelwidth+\labelsep\relax}If $\mathcal{L}_{/i}\neq\{0\}$ for all $i$, then return “No”.
    \end{description}
 (\# Here, assume that the first $r\equals~{\rm rank}(\mathcal{N})$ rows of $\mathcal{N}$ are linearly independent. )
\end{algorithm}

\begin{table}[!t]
\caption{Summary of ACR Detection for Biochemical Networks.}\label{tab:algorithm summary}
\centering
\begin{tabular}{@{}lccccll@{}}
\toprule
Network & $s$ & $r$ & $m$ & Total runtime & Output ACR & Verified ACR\\
\midrule
(a) & 3 & 1 & 2 & 0.03100 & Yes, $Z$ & Yes, $Z$\\
(b) & 3 & 2 & 2 & 0.10900 & Yes, $X$ & Yes, $X$\\
(c) & 4 & 2 & 4 & 0.14000 & No  & No\\
(d) & 5 & 3 & 6 & 0.35900 & Yes, $I$ & Yes, $I$\\
(e) & 7 & 5 & 9 & 1.67200 & Yes, $Y_p$ & Yes, $Y_p$\\
(f) & 8 & 6 & 11 & 4.21900 & Yes, $X_pY$ & Yes, $X_pY$\\
(g) & 8 & 6 & 12 & 9.17200 & Yes, $S_3$ & Yes, $S_3$\\
(h) & 10 & 6 & 12 & 11.01500 & No  & No\\
\bottomrule
\end{tabular}
\footnotetext{\footnotesize \textbf{Note}: (i)  The column "Network" corresponds to the networks illustrated in the online supplementary material at \url{https://github.com/Cynthia-091225/Detect-for-ACR/blob/main/Fig_networks.jpg}.  (ii) The column ``Total runtime'' denotes the runtime (in seconds) of running Algorithm \ref{alg:main}. (iii) The column ``Output ACR'' indicates the output of Algorithm \ref{alg:main}. (iv) The column ``Verified ACR'' indicates the verified results by precise algebraic computation or other classic methods.}
\end{table}

\end{appendices}


\end{document}